\documentclass{article}
\usepackage{amsmath,amsthm,amssymb,cite}
\newtheorem{theorem}{Theorem}[section]

\newtheorem{lemma}[theorem]{Lemma}
\newtheorem{proposition}[theorem]{Proposition}
\theoremstyle{definition}
\newtheorem{definition}[theorem]{Definition}
\theoremstyle{remark}
\newtheorem{remark}[theorem]{Remark}

\DeclareMathOperator*{\esssup}{\mathrm{ess\,sup}}
\DeclareMathOperator*{\essinf}{\mathrm{ess\,inf}}
\allowdisplaybreaks[3]
\numberwithin{equation}{section}

\begin{document}

\title{Markov Evolution of Continuum Particle Systems
    with Dispersion and Competition}

\author{Dmitri Finkelshtein\thanks{Institute of Mathematics,
     National Academy of Sciences of Ukraine,
     01601 Kiev-4, Ukraine, e-mail: fdl@imath.kiev.ua} \and
     Yuri Kondratiev\thanks{Fakult\"at f\"ur Mathematik, Universit\"at Bielefeld,
     Postfach 110 131, 33501 Bielefeld, Germany, e-mail:
     kondrat@math.uni-bielefeld.de} \and
    Yuri Kozitsky\thanks{Instytut Matematyki, Uniwersytet Marii Curie-Sk{\l}odwskiej,
     20-031 Lublin, Poland, e-mail:
        jkozi@hektor.umcs.lublin.pl} \and Oleksandr Kutoviy\thanks{Fakult\"at f\"ur Mathematik, Universit\"at Bielefeld,
     Postfach 110 131, 33501 Bielefeld, Germany, e-mail:
     kutoviy@math.uni-bielefeld.de}}
%\today

 \maketitle

\begin{abstract}
We construct birth-and-death Markov evolution of states
(distributions) of point particle systems in $\mathbb{R}^d$. In this
evolution, particles reproduce themselves at distant points
(disperse) and die under the influence of each other (compete). The
main result is a statement that the corresponding correlation
functions evolve in a scale of Banach spaces and remain
sub-Poissonian, and hence no clustering occurs, if the dispersion is
subordinate to the competition.
\end{abstract}

{\bf Keywords}: Markov evolution, Markov generator, stochastic
semigroup, sun-dual semigroup, spatial ecology, individual-based
model, correlation function, observable.

{\bf MSC (2010)}: 92D25, 60J80, 82C22

\section{Introduction and Overview}
\subsection{Introduction}

Dynamics of interacting particle systems distributed over discrete
regular sets, such as $\mathbb{Z}^d$, has been studied in great
detail, see, e.g., \cite{Durrett,Durrett1,Liggett,Penrose,Spohn}.
However, in many real world applications the underlying space should
essentially be continuous. In this paper, we study Markov evolution
of birth-and-death type of an infinite system of point particles
distributed over $\mathbb{R}^d$, $d\geq 1$. The particles reproduce
themselves, compete and die. The reproduction consists in random
(independent) sending by a particle located at $x$ an offspring to
point $y$, which immediately after that becomes a member of the
system. This process is described by a {\it dispersal kernel},
$a_{+}(x,y) \geq0$. Each particle  dies independently with constant
mortality rate $m\geq0$, as well as under the influence of the rest
of particles (density dependent mortality). The latter process is
described by a {\it competition kernel}, $a_{-} (x,y)\geq0$. The
kernels $a_{\pm}$ determine the corresponding rates in an additive
way. For instance, for the particle located at $x$, the overall
density dependent mortality rate is $\sum_{y}a_{-} (x,y)$, where the
sum is taken over all other particles. Systems of this kind are used
as individual-based models of large ecological communities of
entities (e.g., perennial plants) distributed over a continuous
space (habitat) and evolving in continuous time, cf. page 1311 in
\cite{Neuhauser}. They attract considerable attention of both
mathematicians and theoretical biologists, see, e.g.,
\cite{DimaN1,DimaN2,DimaN,Dima,Dima2,FM,Oles1,KKP,KS,Lach5} and
\cite{DY,BP1,BP2,BP3,DL,Mu}, respectively.

In theoretical biology, models of this kind were described by means
of `spatial moment equations', cf. \cite{BP1,BP2}. These are chains
of (linear) evolutional equations describing the time evolution of
densities and higher order moments\footnote{These moments correspond
to correlation functions used in this article.}, representing
`spatial covariances'. These equations involve the dispersal and
mortality rates mentioned above. The main difficulty encountered by
the authors of those and similar works is that the mentioned chains
are not closed, e.g., the time derivative of the density is
expressed through the two-point moment, whereas the time derivative
of the two-point moment is expressed through the moments of higher
order, etc. Typically, this difficulty is circumvented by means of a
`moment closure' ansatz, in which one approximates moments of order
higher than a certain value by the products of lower order moments,
e.g., the two-point moment is set to be the product of two
densities, and thereby the two-point covariances are neglected. The
equations obtained in this way are closed but nonlinear. An example
can be the Lotka--Volterra equations with spatial dependence derived
in \cite{BP2}. Similar equations are deduced from the microscopic
theory of systems of interacting particles living on $\mathbb{Z}^d$,
see e.g., equation (2) on page 137 in \cite{Durrett}. We refer to
\cite{Durrett,Durrett1,Neuhauser,Perthame} for more details and
references on this topic. At the same time, plenty of such equations
of population biology appear phenomenologically, without employing
individual-based models, cf. \cite{zhao}.

As is well-understood now, the mentioned nonlinear equations play
the role of kinetic equations known in the statistical theory of
Hamiltonian dynamical systems in the continuum. Nowadays, the latter
equations are derived from microscopic equations by means of scaling
procedures, see the corresponding discussion in
\cite{maslov,Dob,P,spohn,Spohn} for more detail. We believe that
also in population biology and other life sciences, the use of
individual-based models with continuous habitat will provide an
adequate mathematical framework for describing the collective
behavior of large systems of interacting entities. In that we mean
the following program realized in part in the present paper. The
states of the model are probability measures on the space of
configurations of particles in $\mathbb{R}^d$. Their dynamics is
described as a Markov evolution by means of
Fokker-Planck-Kolmogorov-type equations.  As in the case of
Hamiltonian systems, this evolution can be constructed by means of
the evolution for the corresponding correlation functions. The
mesoscopic description, which neglects certain features of the model
behavior, is then obtained by a scaling procedure. In its framework,
one studied the scaling limit $\varepsilon \to 0$, in which the
correlation functions converge to `mesoscopic' correlation
functions. By virtue of the scaling procedure, the evolution of the
latter functions `preserves chaos', which means that at each instant
of time such functions are the products of density functions if the
initial correlation functions possess this property. The
corresponding kinetic equation is then the equation for the density
function. Typically, this is a nonlinear and nonlocal
equation\footnote{Cf. the discussion in \cite{Furter}.}. We refer
the reader to \cite{DimaN1,DimaN} for more detail. Generally
speaking, the aim of the present paper is to go further in
developing the micro- and mesoscopic descriptions of the model
mentioned above comparing to what was done in \cite{DimaN2,Dima}. A
more specific presentation of our aims and the results obtained in
this article is given in the next subsection, see also the
concluding remarks in Section~\ref{SCon}.

As was suggested already in \cite{BP1}, the right mathematical
context for studying individual-based models of ecological systems
is the theory of random point fields in $\mathbb{R}^d$, cf. also
page 1311 in \cite{Neuhauser}. Herein, populations are modeled as
particle configurations constituting the configuration space
\begin{equation} \label{C1}
 \Gamma\equiv\Gamma(\mathbb R^d):=\{\gamma\subset\mathbb R^d :
 |\gamma\cap K|<\infty\text{ for any compact $K\subset\mathbb R^d$
 }\},
\end{equation}
where $|A|$ stands for the cardinality of $A$.
Noteworthy, along with finite ones $\Gamma$ contains also infinite
configurations, which allows for describing `bulk' properties of a
large finite system ignoring boundary and size effects\footnote{A
discussion on how infinite systems provide approximations for large
finite systems see, e.g., \cite{Cox,Do}.}. Note that if the initial configuration
$\gamma_0$ is fixed, the evolution might be described as a map
$t\mapsto \gamma_t \in \Gamma$, which in view of the random
character of the events mentioned above ought to be a random
process. However, at least so far, for the model considered here
this way can be realized only if $\gamma_0$ is finite \cite{FM}. On the other hand,
the statistical description of infinite interacting particle systems plays a fundamental role in modern mathematical physics and applications, see, e.g., \cite{Dob}. Namely, the evolution of infinite system should be considered in
terms of dynamics of probability measures (states) on $\Gamma$ rather than point-wisely. To characterize them
one employs {\it observables}, which are appropriate functions
$F:\Gamma \rightarrow \mathbb{R}$. The quantity
\begin{equation*}
 \langle \! \langle F, \mu \rangle \! \rangle = \int_{\Gamma} F (\gamma) \mu(d \gamma)
\end{equation*}
is called the value of observable $F$ in state $\mu$. Then the
system evolution might be described as the evolution of observables
obtained from the Kolmogorov equation
\begin{equation}
 \label{R2}
\frac{d}{dt} F_t = L F_t , \qquad F_t|_{t=0} = F_0, \qquad t>0,
\end{equation}
where the `generator' $L$ specifies the model. The evolution of
states is obtained from the Fokker--Planck equation
\begin{equation}
 \label{R1}
\frac{d}{dt} \mu_t = L^* \mu_t, \qquad \mu_t|_{t=0} = \mu_0,
\end{equation}
related to (\ref{R2}) by the duality
\begin{equation*}
\langle \! \langle F_0, \mu_t \rangle \! \rangle = \langle \!
\langle F_t, \mu_0 \rangle \! \rangle.
\end{equation*}
However, for the model considered in this article, even the mere
definition of the equations in (\ref{R2}) and (\ref{R1}) in
appropriate spaces is rather impossible as the phase space $\Gamma$
contains infinite configurations. Following classical works on the
Hamiltonian dynamics \cite{maslov,Dob,spohn} one can try to study
the evolution of states $\mu_0 \mapsto \mu_t$ via the evolution of
the corresponding correlation functions. For a measure $\mu$ and a
bounded measurable $\Lambda \subset \mathbb{R}^d$, the probability
that in state $\mu$ there is $m$ particles in $\Lambda$ can be
expressed through the $n$-point correlation functions $k^{(n)}_\mu$
with $n\geq m$, see, e.g., \cite{Ruelle}. In particular, for the
Poisson measure $\pi_\varkappa$, see subsection~\ref{SSSec212}
below, $k_{\pi_\varkappa}^{(n)}(x_1 , \dots , x_n) =\varkappa^n$ for
all $n\in \mathbb{N}_0$.

In general, the correlation function $k_\mu$ is a collection of
symmetric functions $k_\mu^{(n)}:(\mathbb{R}^d)^n \to \mathbb{R}$,
$n\in \mathbb{N}_0$, and $k_\mu^{(1)}$ is the particle density. The
evolution $k_0\mapsto k_t$ is obtained from the equation
\begin{equation}
 \label{R4}
\frac{d}{dt} k_t = L^\Delta k_t, \qquad k_t|_{t=0} = k_0,
\end{equation}
which, in fact, is a chain of equations\footnote{For Hamiltonian
systems, the analog of (\ref{R4}) is known as the BBGKY hierarchy.}
for particular $k_t^{(n)}$. Here $L^\Delta$ is constructed from $L$
as in (\ref{R2}) by a certain procedure. According to our program,
the microscopic description consists in proving the existence of
solutions of (\ref{R4}) in appropriate Banach spaces, and in
studying their properties. The next step is to show that these
solutions are correlation functions for some states, which can be
done by showing that they obey certain bounds and are positive
definite in a certain sense, see Proposition~\ref{rhopn}. An
important property of $k_t$ is being {\it sub-Poissonian}, which
means that, for some $C_t>0$, each $k_t^{(n)}$ is bounded by
$C_t^n$. This property can be guaranteed by the appropriate choice
of the Banach space in which one solves (\ref{R4}). Note that the
increase of $k_t^{(n)}$ with $n$ as $n!$ (see \eqref{nfactorial}
below), would correspond to the formation of clusters due to
dispersion. In the present article, we especially address the
question concerning the role of the competition in preventing such
clustering.

\subsection{The overview}

For the considered model, the `generator' in (\ref{R2}) reads
\begin{eqnarray}\label{R20}
(LF)(\gamma) &=& \sum_{x\in \gamma}\left[ m + E^{-}
 (x, \gamma\setminus x) \right]\left[F(\gamma\setminus x) - F(\gamma) \right]\\[.2cm]
&& + \int_{\mathbb{R}^d} E^{+} (y, \gamma) \left[ F(\gamma\cup y) - F(\gamma) \right]dy, \nonumber
\end{eqnarray}
where
\begin{equation}
 \label{Ra20}
E^{\pm}(x, \gamma) := \sum_{y\in \gamma} a_{\pm} (x,y).
\end{equation}
This is a birth-and-death generator for continuous system
(see one of the pioneering papers \cite{HS1978} and also the references in \cite{DimaN2}) in which the first term corresponds to the death of the particle located at $x$
occurring (a) independently with rate $m\geq 0$, and (b) under the
influence of the other particles in $\gamma$ with rate $E^{-}(x,
\gamma\setminus x)\geq 0$. Here and in the sequel in the
corresponding context, we treat each $x\in \mathbb{R}^d$ also as a
single-point configuration $\{x\}$. Note that $a_{-}(x,y)$ describes the interparticle competition. The
second term in (\ref{R20}) describes the birth of a particle at
$y\in \mathbb{R}^d$ given by the whole configuration $\gamma$ with
rate $E^{+}(y, \gamma)\geq 0$. A~particular case of this model is the continuous contact model \cite{KKP,KS} where $a_{-} \equiv 0$, and hence the competition is
absent, see also \cite{Dima}.

As was mentioned above, one of the main question is to study the Cauchy
problem (\ref{R4}) with the corresponding operator $L^\Delta$ in proper Banach spaces. The main characteristic feature of such
spaces is that they should contain only sub-Poissonian correlation
functions. Note that for the contact model ($a_{-} \equiv 0$), mentioned above, it is
known  \cite{Dima} that
\begin{equation}\label{nfactorial}
\mathrm{const}\cdot n!\, c_t^n\leq k_t^{(n)} (x_1 , \dots , x_n)
\leq \mathrm{const}\cdot n!\, C_t^n,
\end{equation}
where the left-hand inequality holds if all $x_i$ belong to a ball
of small enough radius. Hence, in spite of the fact that $C_t \to 0$
as $t\to +\infty$ if the mortality dominates the dispersion (as in
(\ref{yy12}) below), $k_t$ are definitely not sub-Poissonian if
$a_{-} \equiv 0$. On the contrary, according to Theorem~\ref{2tm} below, if, for some
$\theta>0$, we have that $a_{+} \leq \theta a_{-}$ pointwise, cf.
(\ref{z14}), then (\ref{R4}) has a unique (classical) sub-Poissonian
solution on a bounded time interval. It is worth noting, that a solution of (\ref{R4}) is the correlation function of a probability measure on $\Gamma$ if only
it possesses a certain positivity property. In Theorem~\ref{r1tm},
we show that the solution $k_t$, existing according to
Theorem~\ref{2tm}, has this property if $m$ dominates $a_{+}$ in the
sense of (\ref{yy12}).

The rest of the paper is organized as follows. In
Section~\ref{Sec2}, we introduce a necessary mathematical framework
and give a formal definition of the model. The evolution of
correlation functions is studied in Section~\ref{Sec4}. It is,
however, preceded by the study of certain auxiliary objects, {\em
quasi-observables}, the evolution of which is generated by
$\widehat{L}$ whose dual, in the sense of (\ref{19A}), is
$L^\triangle$. This is done in Section~\ref{Sec3}, where we use a
combination of $C_0$-semigroup techniques in ordered Banach spaces
with an Ovcyannikov-type method, which yields the evolution of
quasi-observables in a scale of Banach spaces on a bounded
time-interval, see Theorem~\ref{1tm}. The main peculiarity of the
evolution $k_0\mapsto k_t$ described in Theorem~\ref{2tm} is that
the corresponding Banach spaces are of $L^\infty$-type, which forced
us to use a combination of $C_0$-semigroups, sun-dual to those from
Section~\ref{Sec3}, with Ovcyannikov's method. In
Section~\ref{Sec4}, in addition to the classical solutions of the
Cauchy problem for correlation functions and quasi-observables, we
also study the dual evolutions defined in (\ref{z46}) and
(\ref{z47}). Similarly to the usual $C_0$-semigroup framework, we
obtain that the classical evolution of correlation functions
coincides with the evolution which is dual to the evolution of
quasi-observables, see Proposition~\ref{newprop}. The results of
Section~\ref{Sec4} are used for proving Theorem~\ref{r1tm}
concerning the dynamics of states. Another ingredient of our study
of the dynamics of states is Lemma~\ref{y3tm} where the dynamics of
local densities is described. Note that the latter evolution might
be extended to the evolution of states supported on finite
configurations, that provides an alternative way of constructing the
evolution $\gamma_0\mapsto\gamma_t$ mentioned above. The concluding
remarks are presented in Section~\ref{SCon}.

In the second part of this work, which will be published as a
separate paper, we perform the following. In the framework of the
scaling approach developed in \cite{DimaN1,DimaN}, we pass to the
following analog of (\ref{R4})
\begin{equation}
  \label{QRA}
\frac{d}{d t} k_{t, {\rm ren}}^{(\varepsilon)} = (V+ \varepsilon
C)k_{t, {\rm ren}}^{(\varepsilon)}, \qquad k_{t, {\rm
ren}}^{(\varepsilon)}|_{t=0} = r_0,
\end{equation}
where $V$ and $C$ are certain operators such that $L^\Delta = V +
C$, and the initial correlation function $r_0$ is such that, for
$n\in \mathbb{N}$,
\begin{equation}
  \label{QR}
 r^{(n)}_0 (x_1 , \dots , x_n) = \varrho_0 (x_1) \cdots \varrho
 (x_n).
\end{equation}
For the problem (\ref{QRA}), the statement of Theorem~\ref{2tm}
holds true for all $\varepsilon \in (0,1]$. Passing to the limit
$\varepsilon \to 0$ in the equation above we arrive at
\begin{equation}
  \label{QR1}
\frac{d}{d t} r_{t} = Vr_{t}, \qquad r_{t}|_{t=0} = r_0.
\end{equation}
For the latter problem, the statement of Theorem~\ref{2tm} also
holds true, even without the restriction imposed in (\ref{z14}).
Next,for the solutions of (\ref{QRA}) and (\ref{QR1}),  we prove
that $k_{t, {\rm ren}}^{(\varepsilon)} \to r_t$ as $\varepsilon \to
0$, which holds uniformly on compact subsets of the time interval
and in the spaces where we solve both problems. The peculiarity of
the equation in (\ref{QR1}) is such that its solution can be sought
in the product form of (\ref{QR}), with the corresponding density
$\varrho_t$. This leads to the following equations
\begin{equation}
  \label{QR2}
\frac{d}{d t} \varrho_{t}(x) = - m \varrho_t(x) +
\int_{\mathbb{R}^d} a_{+}(x,y) \varrho_t (y) dy - \varrho_t (x)
\int_{\mathbb{R}^d} a_{-}(x,y) \varrho_t (y) dy,
\end{equation}
where $a_{\pm}$ are the same as in (\ref{Ra20}). Then we prove that
the above equation has a unique classical solution in a ball in
$L^\infty(\mathbb{R}^d)$ such that $r_t$ expressed as the product of
these $\varrho_t$ is the unique solution of (\ref{QR}). We also find
some interesting properties of the solutions of (\ref{QR2}). Note
that a particular case of (\ref{QR2}) with
\[
a_{+} (x,y) = a_{-} (x,y) = \psi(x-y),
\]
was derived in \cite{Durrett} (a crabgrass model, see also page 1307
in \cite{Neuhauser}). The front propagation in the crabgrass model
was studied in \cite{Perthame}.

\section{The Basic Notions and the Model}\label{Sec2}
\subsection{The notions}

All the details of the mathematical framework of this paper can be
found in \cite{Albev,Dima,Dima2,Tobi,Oles0,KKP,KS,Obata}. Recall
that we consider an infinite system of point particles distributed
over $\mathbb{R}^d$, $d\geq 1$.
 By $\mathcal{B}(\mathbb{R}^d)$ and
 $\mathcal{B}_{\rm b}(\mathbb{R}^d)$ we denote the set of all
Borel and all bounded Borel subsets of $\mathbb{R}^d$, respectively.

\subsubsection{The configuration spaces}

The configuration space $\Gamma$ has been defined in (\ref{C1}).
Each $\gamma \in \Gamma$ can be identified with the following
positive integer-valued Radon measure
\[
\gamma(dx)= \sum_{y\in \gamma} \delta_y (dx) \in \mathcal{M}(\mathbb{R}^d),
\]
where $\delta_y$ is the Dirac measure centered at $y$, and
$\mathcal{M}(\mathbb{R}^d)$ stands for the set of all positive Radon
measures on $\mathcal{B}(\mathbb{R}^d)$. This allows us to consider
$\Gamma$ as the subset of $\mathcal{M}(\mathbb{R}^d)$, and hence to
endow it with the vague topology. By definition, this is the weakest
topology in which all the maps
\[
\Gamma \ni \gamma \mapsto
\int_{\mathbb{R}^d} f(x) \gamma (dx)= \sum_{x\in \gamma} f(x) , \quad f\in C_0 (\mathbb{R}^d),
\]
are continuous. Here $C_0 (\mathbb{R}^d)$ stands for the set of all
continuous functions $f:\mathbb{R}^d \rightarrow \mathbb{R}$ which
have compact supports. The vague topology on $\Gamma$ admits a
metrization which turns it into a complete and separable metric
(Polish) space, see, e.g., Theorem 3.5 in \cite{Oles0}. By
$\mathcal{B}(\Gamma)$ we denote the corresponding Borel
$\sigma$-algebra.

For $n\in\mathbb{N}_0:=\mathbb{N}\cup\{0\}$, the set of $n$-particle
configurations in $\mathbb{R}^d$  is
 \[ \Gamma^{(0)} = \{ \emptyset\}, \qquad \Gamma^{(n)} = \{\eta \subset X: |\eta| = n \}, \ \ n\in \mathbb{N}.
 \]
For $n\geq 2$, $\Gamma^{(n)}$ can be identified with the
symmetrization of the set
\[
\left\{(x_1, \dots , x_n)\in \bigl(\mathbb{R}^{d}\bigr)^n: x_i \neq x_j, \ {\rm for} \ i\neq j\right\},
\]
which allows one to introduce the corresponding topology on
$\Gamma^{(n)}$ and hence the Borel $\sigma$-algebra
$\mathcal{B}(\Gamma^{(n)})$. The set of finite configurations
$\Gamma_0$ is the disjoint union of $\Gamma^{(n)}$, that is,
\begin{equation*}
\Gamma_{0} = \bigsqcup_{n\in \mathbb{N}_0} \Gamma^{(n)} .
\end{equation*}
We endow $\Gamma_0$ with the topology of the disjoint union and
hence with the Borel $\sigma$-algebra $\mathcal{B}(\Gamma_{0})$.
Obviously, $\Gamma_0$ can also be considered as a subset of
$\Gamma$. However, the topology just mentioned and that induced on
$\Gamma_0$ from $\Gamma$ do not coincide.

In the sequel, $\Lambda\subset \mathbb{R}^d$ will always denote a
bounded measurable subset, i.e., $\Lambda \in \mathcal{B}_{\rm
b}(\mathbb{R}^d)$. For such $\Lambda$, we set
\[
 \Gamma_\Lambda = \{ \gamma \in \Gamma: \gamma = \gamma \cap \Lambda\}.
\]

Clearly, $\Gamma_\Lambda$ is a measurable subset of $\Gamma_0$ and
the following holds
\[
\Gamma_\Lambda = \bigsqcup_{n\in \mathbb{N}_0} \Gamma^{(n)}_\Lambda,
\qquad \Gamma^{(n)}_\Lambda:=\Gamma^{(n)} \cap \Gamma_\Lambda,
\]
which allows one to equip $\Gamma_\Lambda$ with the topology induced
by that of $\Gamma_0$. Let $\mathcal{B}(\Gamma_\Lambda)$ be the
corresponding Borel $\sigma$-algebra. It can be proven, see
Lemma~1.1 and Proposition~1.3 in \cite{Obata}, that
\begin{equation*}
\mathcal{B}(\Gamma_\Lambda) = \{ \Gamma_\Lambda \cap \Upsilon : \Upsilon \in \mathcal{B}(\Gamma)\}.
\end{equation*}
Next, we define the projection
\begin{equation}
 \label{K3}
\Gamma \ni \gamma \mapsto p_\Lambda (\gamma)= \gamma_\Lambda :=
\gamma\cap \Lambda, \qquad \Lambda \in \mathcal{B}_{\rm b}
(\mathbb{R}^d).
\end{equation}
It is known, cf. page 451 in \cite{Albev}, that
$\mathcal{B}(\Gamma)$ is the smallest $\sigma$-algebra of subsets of
$\Gamma$ such that the maps $p_\Lambda$ with all $\Lambda \in
\mathcal{B}_{\rm b} (\mathbb{R}^d)$ are
$\mathcal{B}(\Gamma)/\mathcal{B}(\Gamma_\Lambda)$ measurable. This
means that $(\Gamma, \mathcal{B}(\Gamma))$ is the projective limit
of the measurable spaces $(\Gamma_\Lambda,
\mathcal{B}(\Gamma_\Lambda))$, $\Lambda \in \mathcal{B}_{\rm b}
(\mathbb{R}^d)$.

\subsubsection{Measures and functions on configuration spaces}
\label{SSSec212} The basic examples of measures on $\Gamma$ and
$\Gamma_0$ are the Poisson measure $\pi$ and the Lebesgue--Poisson
measure $\lambda$, respectively, cf. Section 2.2 in \cite{Albev}.

The image of the Lebesgue product measure $dx_1 dx_2 \cdots dx_n $ in
$(\Gamma^{(n)}, \mathcal{B}(\Gamma^{(n)}))$ is denoted by $\sigma^{(n)}$.
For $\varkappa>0$, the Lebesgue--Poisson measure on $(\Gamma_0,
\mathcal{B}(\Gamma_0))$ is
\begin{equation}
 \label{1A}
 \lambda_\varkappa := \delta_{\emptyset} + \sum_{n=1}^\infty \frac{\varkappa^n}{n!} \sigma^{(n)}.
\end{equation}
For $\Lambda\in \mathcal{B}_{\rm b}(\mathbb{R}^d)$, the restriction
of $\lambda_\varkappa$ to $\Gamma_{\Lambda}$ will be denoted by
$\lambda^\Lambda_\varkappa$. However, we shall drop the superscript
if no ambiguity arises. Clearly, $\lambda^\Lambda_\varkappa$ is a
finite measure on $\mathcal{B}(\Gamma_\Lambda)$ such that
$\lambda^\Lambda_\varkappa (\Gamma_\Lambda)= e^{
\varkappa\ell(\Lambda)}$, where $\ell(\Lambda)$ is the Lebesgue
measure of $\Lambda$. Then
\begin{equation}
 \label{3A}
 \pi^\Lambda_\varkappa := \exp ( - \varkappa\ell(\Lambda)) \lambda^\Lambda_\varkappa
\end{equation}
is a probability measure on $\mathcal{B}(\Gamma_\Lambda)$. It can be
shown \cite{Albev} that the family
$\{\pi^\Lambda_\varkappa\}_{\Lambda \in \mathcal{B}_{\rm
b}(\mathbb{R}^d)}$ is consistent, and hence there exists a unique
probability measure, $\pi_\varkappa$, on $\mathcal{B}(\Gamma)$ such
that
\[
 \pi^\Lambda_\varkappa = \pi_\varkappa \circ p^{-1}_\Lambda, \qquad \Lambda \in \mathcal{B}_{\rm b} (\mathbb{R}^d),
\]
where $p_\Lambda$ is as in (\ref{K3}). This $\pi_\varkappa$ is called the Poisson measure with intensity $\varkappa>0$. If $\varkappa=1$ we shall drop the subscript and consider the Lebesgue--Poisson measure $\lambda$ and the Poisson measure $\pi$.

Now we turn to functions on $\Gamma_0$ and $\Gamma$.
In fact,
any measurable $G: \Gamma_0 \rightarrow \mathbb{R}$ is a sequence of measurable symmetric functions $G^{(n)} :
(\mathbb{R}^d)^n \rightarrow \mathbb{R}$.
A measurable
$F:\Gamma \rightarrow \mathbb{R}$ is called a cylinder
function if there exist $\Lambda \in \mathcal{B}_{\rm
b}(\mathbb{R}^d)$ and a measurable $G: \Gamma_{\Lambda}\rightarrow
\mathbb{R}$ such that, cf. (\ref{K3}), $F(\gamma) = G(\gamma_\Lambda)$ for all
$\gamma \in \Gamma$. By $\mathcal{F}_{\rm cyl}(\Gamma)$ we denote
the set of all cylinder functions.

A set $\Upsilon\in \mathcal{B}(\Gamma_0)$ is said to be bounded if
\begin{equation}
 \label{6A}
 \Upsilon \subset \bigsqcup_{n=0}^N \Gamma^{(n)}_\Lambda
\end{equation}
for some $\Lambda \in \mathcal{B}_{\rm b}(\mathbb{R}^d)$ and $N\in
\mathbb{N}$. By $B_{\rm bs} (\Gamma_0)$ we denote the set of all
bounded measurable functions $G: \Gamma_0 \rightarrow \mathbb{R}$,
which have bounded supports. That is, each such $G$ is the zero
function on $\Gamma_0 \setminus \Upsilon$ for some bounded $\Upsilon$.
For $\gamma \in \Gamma$, by
writing $\eta \Subset \gamma$ we mean that $\eta \subset \gamma$ and
$\eta$ is finite, i.e., $\eta \in \Gamma_0$. For $G \in B_{\rm
bs}(\Gamma_0)$, we set
\begin{equation}
 \label{7A}
 (KG)(\gamma) = \sum_{\eta \Subset \gamma} G(\eta), \qquad \gamma \in \Gamma.
\end{equation}
Obviously, $K$ is a linear and positivity preserving map, which maps
$B_{\rm bs}(\Gamma_0)$ into $\mathcal{F}_{\rm cyl}(\Gamma)$, see,
e.g., \cite{Tobi}. In the sequel, we use the following set
\begin{equation}
 \label{9AY}
 B_{\rm bs}^+(\Gamma_0) := \{ G \in B_{\rm bs}(\Gamma_0): KG \not\equiv 0, \quad  (KG)(\gamma) \geq 0 \ \ \ {\rm for} \ \ {\rm all} \ \ \gamma \in \Gamma\}.
\end{equation}
By $\mathcal{M}^1_{\rm fm} (\Gamma)$ we denote the set of all
probability measures on $\mathcal{B}(\Gamma)$ that have finite
local moments, that is, for which
\begin{equation*}
 \int_\Gamma |\gamma_\Lambda|^n \mu(d \gamma) < \infty, \qquad \text{for all $n\in \mathbb{N}$ and $\Lambda \in \mathcal{B}_{\rm b} (\mathbb{R}^d)$}.
\end{equation*}
A measure $\mu \in \mathcal{M}^1_{\rm fm} (\Gamma)$
is said to be {\it locally absolutely continuous} with respect
to the Poisson measure $\pi$ if, for every
$\Lambda \in \mathcal{B}_{\rm b} (\mathbb{R}^d)$, $\mu^\Lambda := \mu \circ p_\Lambda^{-1}$
is absolutely continuous with respect to $\pi^\Lambda$, cf.~(\ref{3A}).
A measure $\rho$ on $(\Gamma_0, \mathcal{B}(\Gamma_0))$ is said
to be {\it locally finite} if $\rho(\Upsilon)< \infty$ for every bounded
$\Upsilon\subset \Gamma_0$. By $\mathcal{M}_{\rm lf} (\Gamma_0)$ we denote the set of all such measures.
For a bounded $\Upsilon\subset \Gamma_0$, let $\mathbb{I}_\Upsilon$ be its
indicator function. Then $\mathbb{I}_\Upsilon$ is in $B_{\rm
bs}(\Gamma_0)$ and hence one can apply (\ref{7A}). For $\mu \in
\mathcal{M}^1_{\rm fm} (\Gamma)$, the representation
\begin{equation}
 \label{9A}
 \rho_\mu (\Upsilon) = \int_{\Gamma} (K\mathbb{I}_\Upsilon) (\gamma) \mu(d \gamma)
\end{equation}
determines a unique measure $\rho_\mu \in \mathcal{M}_{\rm
lf}(\Gamma_0))$. It is called the {\it correlation measure} for
$\mu$. Then (\ref{9A}) defines the map $K^*: \mathcal{M}^1_{\rm fm} (\Gamma)
\rightarrow \mathcal{M}_{\rm lf}(\Gamma_0)$ such that $K^*\mu =
\rho_\mu$. In particular, $K^*\pi = \lambda$. It is known,
see Proposition 4.14 in \cite{Tobi}, that $\rho_\mu $ is absolutely
continuous with respect to $\lambda$ if $\mu$ is locally absolutely
continuous with respect to $\pi$. In this case, for any
$\Lambda \in \mathcal{B}_{\rm b} (\mathbb{R}^d)$, we have that
\begin{eqnarray}
 \label{9AA}
 k_\mu (\eta) = \frac{d \rho_\mu}{d \lambda}(\eta) &=&
\int_{\Gamma_{\Lambda}} \frac{d\mu^\Lambda}{d \pi^\Lambda} (\eta \cup \gamma) \pi^\Lambda (d \gamma)\\[.2cm]
&=& \int_{\Gamma_{\Lambda}} \frac{d\mu^\Lambda}{d \lambda^\Lambda} (\eta \cup \gamma) \lambda^\Lambda (d \gamma).\label{9AAghf}
\end{eqnarray}
The Radon--Nikodym derivative $k_\mu$ is called the {\it correlation
function} corresponding to the measure $\mu$. In the sequel, we
shall tacitly assume that the equalities or inequalities, like
(\ref{9AAghf}) or (\ref{9AW}), hold for $\lambda$-almost all $\eta
\in \Gamma_0$. The following fact is known, see Theorems 6.1 and 6.2
and Remark 6.3 in \cite{Tobi}.
\begin{proposition}
 \label{rhopn}
Let $\rho\in \mathcal{M}_{\rm lf}(\Gamma_0)$ have the following properties:
\begin{equation}
 \label{9AZ}
 \rho\left( \emptyset \right) = 1, \qquad \ \ \int_{\Gamma_0} G(\eta) \rho(d\eta) \geq 0 \quad {\rm for} \ \ {\rm all} \ \ G\in B_{\rm bs}^+(\Gamma_0).
\end{equation}
Then there exist $\mu\in \mathcal{M}_{\rm fm}^1(\Gamma)$ such that $K^* \mu = \rho$. Such $\mu$ is unique if
\begin{equation}
 \label{9AW}
 \frac{d\rho}{d \lambda} (\eta) \leq \prod_{x\in \eta} C (x),\qquad \eta\in\Gamma_0,
\end{equation}
for some locally integrable $C:\mathbb{R}^d \to \mathbb{R}_+$.
\end{proposition}
Here and below we use the conventions
\[
\sum_{a\in \emptyset} \phi_a := 0 , \qquad
\prod_{a\in \emptyset} \psi_a := 1.
\]
Finally, we mention the following integration rule, see, e.g., \cite{Dima},
\begin{equation}
 \label{12A}
\int_{\Gamma_0} \sum_{\xi \subset \eta} H(\xi, \eta \setminus \xi,
\eta) \lambda (d \eta) = \int_{\Gamma_0}\int_{\Gamma_0}H(\xi, \eta,
\eta\cup \xi) \lambda (d \xi) \lambda (d \eta),
\end{equation}
 which holds for any appropriate function $H$ if both sides are finite.

\subsection{The model}

An informal generator corresponding to the model is given in (\ref{R20}). The competition and dispersion rates $E^{\pm}(x, \gamma)$ are supposed to be additive, and the corresponding kernels $a_{\pm}$ are translation invariant, see \cite{BP1}. In view of the latter assumption, we write them as
\begin{equation*}
a_{\pm} (x,y) = a^{\pm} (x-y),
\end{equation*}
and hence, cf. (\ref{Ra20}),
\begin{equation}
 \label{13A}
E^{\pm}(x, \gamma) = \sum_{y\in \gamma} a^{\pm} (x-y).
\end{equation}
We suppose that
\begin{equation}
 \label{AA}
a^{\pm} \in L^1(\mathbb{R}^d)\cap L^\infty(\mathbb{R}^d), \qquad a^{\pm}(x) = a^{\pm}(-x) \geq 0,
\end{equation}
and thus set
\begin{equation}
 \label{14A}
\langle{a}^{\pm}\rangle = \int_{\mathbb{R}^d} a^{\pm} (x)dx,\qquad \|a^{\pm} \|= \esssup_{x\in \mathbb{R}^d} a^{\pm}(x),
\end{equation}
and
\begin{equation}
 \label{15A}
E^{\pm} (\eta) = \sum_{x\in \eta}E^{\pm} (x,\eta\setminus x) = \sum_{x\in \eta} \sum_{y\in \eta\setminus x}a^{\pm} (x-y), \quad
\eta \in \Gamma_0 .
\end{equation}
By (\ref{AA}), we have
\begin{equation}
 \label{AB}
E^{\pm} (\eta) \leq \|a^{\pm} \| |\eta|^2.
\end{equation}
For the sake of brevity, we also denote
\begin{equation}
 \label{16A}
E(\eta) = \sum_{x\in \eta}\left(m + E^{-} (x, \eta\setminus x) \right) = m|\eta| + E^{-}(\eta),
\end{equation}
where $m$ is the same as in (\ref{R20}).

Following the general scheme developed in \cite{Oles1} one
constructs the evolution of correlation functions as a dual evolution to that of {\it quasi-observables}, which are
functions $G:\Gamma_0\to \mathbb{R}$. This latter evolution is obtained from
the following Cauchy problem
\begin{equation}
 \label{21}
  \frac{d}{dt} G_t (\eta) = \widehat{L}G_t (\eta), \qquad G_t|_{t=0} = G_0,
\end{equation}
where
\begin{equation}
 \label{22}
 \widehat{L} = K^{-1} L K
\end{equation}
is the so called {\it symbol} of $L$, which has the form, cf. \cite{Dima},
\begin{equation}
 \label{z}
\widehat{L} = A+B
\end{equation}
with
\begin{gather}\label{z1}
A = A_1 + A_2 \\[.2cm]
(A_1 G) (\eta) = - E(\eta) G(\eta), \quad (A_2 G) (\eta) = \int_{\mathbb{R}^d} E^{+} (y,\eta) G (\eta\cup y) dy, \label{z111}
\end{gather}
and
\begin{eqnarray}
 \label{z2}
B &=& B_1 + B_2 ,\\[.2cm]
(B_1 G)(\eta)&=& - \sum_{x\in \eta} E^{-} (x,\eta\setminus x) G(\eta\setminus x),\nonumber \\
(B_2 G)(\eta)&=& \int_{\mathbb{R}^d} \sum_{x\in \eta} a_{+} (x,y) G(\eta \setminus x \cup y) d y. \nonumber
\end{eqnarray}
Clearly, the action of $\widehat{L}$ on $G\in B_{\rm bs}(\Gamma_0)$ is well-defined. Its extension
to wider classes of $G$ will be done in Section~\ref{Sec3} below.

For a measurable locally integrable function $k:\Gamma_0 \rightarrow \mathbb{R}$ and
$G \in B_{\rm bs} (\Gamma_0)$, we define
\begin{equation}
 \label{19A}
 \langle \! \langle G,k\rangle \! \rangle = \int_{\Gamma_0} G(\eta) k(\eta)\lambda (d \eta) .
\end{equation}
This pairing can be extended to appropriate classes of $G$ and $k$.
Then the Cauchy problem `dual' to
(\ref{21}) takes the form
\begin{equation}\label{20A}
\frac{d k_t}{d t} = L^\Delta k_t , \qquad k_t|_{t=0} = k_0,
\end{equation}
where the action of $L^\Delta$ is obtained by means of (\ref{12A}) according to the rule
\[
\langle \! \langle \widehat{L} G, k \rangle \! \rangle = \langle \!
\langle G, L^\Delta k \rangle \! \rangle,
\]
as well as from (\ref{19A}) and (\ref{z})--(\ref{z2}). It thus has the form, cf. \cite{Dima},
\begin{equation}
{L}^\Delta = A^\Delta+B^\Delta \label{21A}
\end{equation}
with
\begin{gather}
A^\Delta = A^\Delta_1 + A^\Delta_2 \label{22A}\\[.2cm]
(A_1 k)(\eta) = - E(\eta) k(\eta), \qquad (A_2 k)(\eta) = \sum_{x\in \eta} E^{+} (x,\eta\setminus x)
k (\eta\setminus x), \nonumber
\end{gather}
and
\begin{eqnarray}
B^\Delta &= & B^\Delta_1 + B^\Delta_2 , \label{23A}\\[.2cm]
(B^\Delta_1 k)(\eta) &=& - \int_{\mathbb{R}^d} E^{-} (y,\eta) k(\eta\cup y) dy ,\nonumber \\[.2cm] (B^\Delta_2 k)(\eta) &=&
 \int_{\mathbb{R}^d} \sum_{x\in \eta} a^{+} (x - y) k(\eta \setminus x \cup y) d y. \nonumber
\end{eqnarray}
Of course, like $\widehat{L}$ the above introduced $L^\Delta$ is well-defined only for `good enough' $k$. In the next sections,
we define both operators in the corresponding Banach spaces.

\section{The Evolution of Quasi-observables}\label{Sec3}

\subsection{Setting}

For $\alpha \in \mathbb{R}$ and the measure $\lambda$ as in (\ref{1A}), we consider the Banach space
\begin{equation}
 \label{z3}
\mathcal{G}_\alpha := L^1 (\Gamma_0 , e^{-\alpha |\cdot|} d \lambda),
\end{equation}
in which the norm is
\begin{equation*}
\|G\|_{\alpha} = \int_{\Gamma_0} |G(\eta)| \exp(- \alpha |\eta|) \lambda (d \eta).
\end{equation*}
Clearly, $\|G\|_{\alpha'} \leq \|G\|_{\alpha''}$ for $\alpha'' < \alpha'$; hence, we have that
\begin{equation}
 \label{z5}
\mathcal{G}_{\alpha''} \hookrightarrow \mathcal{G}_{\alpha'}, \qquad {\rm for} \ \ \alpha'' < \alpha',
\end{equation}
where the embedding is dense and continuous. Now we fix $\alpha\in \mathbb{R}$ and turn to the definition of
$\widehat{L}$ in $\mathcal{G}_\alpha$, see (\ref{z})--(\ref{z2}). Set
\begin{eqnarray*}
\mathcal{D}(A_1)& = & \{G\in \mathcal{G}_\alpha : E(\cdot) G(\cdot) \in \mathcal{G}_\alpha\},\\[.2cm]
\mathcal{D}(A_2)& = & \{G\in \mathcal{G}_\alpha : E^{+}(\cdot) G(\cdot) \in \mathcal{G}_\alpha\},
\end{eqnarray*}
where $E^{\pm}(\eta)$ are as in (\ref{15A}).
As a multiplication operator, $A_1$ with ${\rm Dom}(A_1) = \mathcal{D}(A_1)$ is closed. By (\ref{12A}), for an appropriate $G$, we get
\begin{eqnarray}
 \label{z7}
\|A_2 G\|_\alpha & \leq & \int_{\Gamma_0} \int_{\mathbb{R}^d} E^{+}(y, \eta) |G(\eta\cup y)|e^{-\alpha|\eta|}dy \lambda(d\eta)\\[.2cm]
& = & e^\alpha \int_{\Gamma_0} |G(\eta)| e^{-\alpha|\eta|} \left(\sum_{x\in \eta} E^{+} (x, \eta \setminus x) \right) \lambda ( d \eta)
\nonumber\\[.2cm]
& = & e^\alpha \|E^{+} (\cdot) G(\cdot)\|_\alpha .\nonumber
\end{eqnarray}
Hence, $A_2$ with ${\rm Dom}(A_2) = \mathcal{D}(A_2)$ is well-defined. Further, we set
\begin{equation*}
\mathcal{D}(B) = \{G\in \mathcal{G}_\alpha : |\cdot| G(\cdot) \in \mathcal{G}_\alpha\}.
\end{equation*}
Like in (\ref{z7}), for an appropriate $G$, we obtain
\begin{eqnarray}
\label{z9}
\|B_1 G\|_\alpha & \leq & \int_{\Gamma_0}\sum_{x\in \eta} E^{-} (x, \eta \setminus x) |G(\eta \setminus x)|
 e^{-\alpha |\eta|} \lambda (d \eta) \\[.2cm]
& = & e^{-\alpha} \int_{\Gamma_0} \left(\int_{\mathbb{R}^d} E^{-} (y, \eta)dy \right)|G(\eta)| e^{-\alpha |\eta|}
\lambda(d\eta) \nonumber\\[.2cm]
& = & e^{-\alpha} \langle{a}^{-}\rangle \int_{\Gamma_0} |\eta| |G(\eta)| e^{-\alpha |\eta|}
\lambda(d\eta), \nonumber
\end{eqnarray}
where we have used (\ref{14A}). In a similar way, we get
\begin{equation}
\label{z10}
\|B_2 G\|_\alpha \leq \langle{a}^{+} \rangle \int_{\Gamma_0} |\eta| |G(\eta)| e^{-\alpha |\eta|}
\lambda(d\eta).
\end{equation}
Thus, the operator $B$ as in (\ref{z2}) with ${\rm Dom}(B)= \mathcal{D}(B)$ is also well-defined. Thereafter,
we set
\begin{equation}
 \label{z11}
{\rm Dom}(\widehat{L}) = \mathcal{D}(A_1) \cap \mathcal{D}(A_2)\cap \mathcal{D}(B).
\end{equation}
For $\varkappa>0$ and any $\eta \in \Gamma_0$, we have that
\begin{equation}
 \label{z11Aa}
 |\eta| e^{-\varkappa |\eta|} \leq \frac{1}{ e \varkappa}, \qquad \quad |\eta|^2 e^{-\varkappa |\eta|} \leq \left(\frac{2}{ e \varkappa}\right)^2.
\end{equation}
Then
 by (\ref{z9}) and (\ref{z10})
\begin{equation}
 \label{z13}
\|B G\|_\alpha \leq \frac{\langle{a}^{+} \rangle + \langle{a}^{-}\rangle e^{-\alpha}}{e(\alpha - \alpha')}\|G\|_{\alpha'},
\end{equation}
which holds for any $\alpha' < \alpha$. By the second estimate in (\ref{z11Aa}), and
by (\ref{AB}) and (\ref{16A}), we also get
\begin{equation}
 \label{z11a}
\esssup_{\eta \in \Gamma_0} E(\eta)\exp(- \varkappa|\eta|) \leq M'/\varkappa^2,
\quad \esssup_{\eta \in \Gamma_0} E^{+}(\eta)\exp(- \varkappa|\eta|) \leq M''/\varkappa^2,
\end{equation}
which holds for any $\varkappa >0$ and some positive $M'$ and $M''$.
Thus, we have proven the following
\begin{lemma}
 \label{LLpn}
For each $\alpha' < \alpha$, the expressions (\ref{z})--(\ref{z2}) define a bounded linear operator acting from $\mathcal{G}_{\alpha'}$
into $\mathcal{G}_\alpha$, which we also
denote by $\widehat{L}$, such that the corresponding operator norm obeys the estimate
\begin{equation}
 \label{z11Ab}
 \|\widehat{L} \|_{\alpha'\alpha} \leq M/(\alpha - \alpha')^2,
\end{equation}
for some $M>0$. Furthermore, the same expressions and (\ref{z11}) define an unbounded operator on $\mathcal{G}_\alpha$ such that, for any $\alpha' < \alpha$,
\begin{equation}
 \label{z12}
\mathcal{G}_{\alpha'}\subset{\rm Dom}(\widehat{L}).
\end{equation}
\end{lemma}
\begin{definition}
By a classical solution of the problem (\ref{21}), in the space $\mathcal{G}_\alpha$ and on the time interval $[0,T)$, we understand a map $[0,T)\ni t \mapsto G_t \in {\rm Dom}(\widehat{L}) \subset\mathcal{G}_\alpha$, continuous on $[0,T)$ and continuously differentiable on $(0,T)$, such that (\ref{21}) is satisfied for $t\in[0,T)$.
\end{definition}
\begin{remark}
In view of (\ref{z12}), the condition $G_t \in {\rm Dom}(\widehat{L})$ can be verified by showing that the solution $G_{t}$ belongs to $\mathcal{G}_{\alpha_t}$ for some $\alpha_t < \alpha$.
\end{remark}

\subsection{The statement}

The basic assumption regarding the model properties which we need is
the following: there exists $\theta >0$ such that, for almost all
$x\in \mathbb{R}^d$,
\begin{equation}
 \label{z14}
a^{+} (x) \leq \theta a^{-} (x).
\end{equation}
For $\alpha^*\in \mathbb{R}$ and $\alpha_* < \alpha^*$, we set
\begin{equation}
 \label{z15}
T_* = \frac{\alpha^* - \alpha_*}{\langle{a}^{+}\rangle + \langle{a}^{-}\rangle e^{-\alpha_*}}.
\end{equation}
\begin{theorem}
 \label{1tm}
Let (\ref{z14}) be satisfied. Then, for every $\alpha^* \in \mathbb{R}$ such that
\begin{equation}
 \label{z16}
e^{\alpha^*} \theta < 1,
\end{equation}
and any $\alpha_*<\alpha^*$, the problem (\ref{21}) with $G_0 \in
\mathcal{G}_{\alpha_*}$ has a unique classical solution in
$\mathcal{G}_{\alpha^*}$ on the time interval $[0, T_*)$ with $T_*$
given in (\ref{z15}). \end{theorem} The main idea of the proof is to
obtain the solution as the limit in $\mathcal{G}_{\alpha^*}$ of the
sequence $\{G^{(n)}_t\}_{n\in \mathbb{N}_0}$ which we obtain
recursively by solving the following Cauchy problems
\begin{equation}
 \label{z17}
\frac{d}{dt} G^{(n)}_t = AG^{(n)}_t +B G^{(n-1)}_t, \qquad G^{(n)}_t|_{t=0} = G_0, \ \ n\in \mathbb{N},
\end{equation}
and
\begin{equation}
 \label{z18}
\frac{d}{dt} G^{(0)}_t = AG^{(0)}_t, \qquad G^{(0)}_t|_{t=0} = G_0,
\end{equation}
where $A$ and $B$ are given in (\ref{z1}) and (\ref{z2}),
respectively. The reason to split $\widehat{L}$ as in (\ref{z}) is
the following. In view of (\ref{z13}), $B$ acts continuously from a
smaller $\mathcal{G}_{\alpha''}$ into a bigger
$\mathcal{G}_{\alpha'}$, cf. (\ref{z5}). The fact that the
denominator in (\ref{z13}) contains the difference $\alpha-\alpha'$
in the power one allows for employing Ovcyannikov's type arguments,
see, e.g., \cite{trev}. However, this is true only for $B$ but not
for $A$, cf. (\ref{z11a}) and (\ref{z11Ab}). In Lemma~\ref{1lm}
below, we prove that under the assumption (\ref{z16}) $A$~is the
generator of a substochastic analytic semigroup\footnote{Which is
the only reason for imposing (\ref{z16}).}. Combining these facts
and employing standard results of the theory of inhomogeneous
differential equations in Banach spaces, we prove the existence of
$G^{(n)}_t$, $n\in \mathbb{N}_0$, and then the convergence
$G^{(n)}_t \to G_t$. The uniqueness is proven by showing that the
only classical solution of the problem (\ref{21}) with the zero
initial condition is $G_t \equiv 0$.

In the proof of Lemma~\ref{1lm} below we employ the perturbation
theory for positive semigroups of operators in ordered Banach spaces
developed in \cite{TV}. Prior to stating this lemma we present the
relevant fragments of this theory in the special case of spaces of
integrable functions. Let $E$ be a~measurable space with
a~$\sigma$-finite measure $\nu $, and $X:=L^{1}\left( E\rightarrow
\mathbb{R},d\nu \right) $ be the Banach space of $\nu $-integrable
real-valued functions on $X$ with norm $\left\Vert \cdot \right\Vert
$. Let $X^{+}$ be the cone in $X$ consisting of all $\nu $-a.e.
nonnegative functions on $E$. Clearly, $\left\Vert f+g\right\Vert
=\left\Vert f\right\Vert +\left\Vert g\right\Vert $ for any $f,g\in
X^{+}$, and this cone is generating, that is, $X=X^{+}-X^{+}$.
Recall that a $C_0$-semigroup $\{S(t)\}_{t\geq0}$ of bounded linear
operators on $X$ is called \emph{positive} if $S(t)f\in X^+$ for all
$f\in X^+$. A positive semigroup is called \emph{substochastic}
(corr., \emph{stochastic}) if $\|S(t)f\| \leq \|f\|$ (corr.,
$\|S(t)f\| = \|f\|$) for all $f\in X^+$. Let $\left(
A_{0},D(A_{0})\right) $ be the generator of a positive $C_{0}$
-semigroup $\{S_{0}\left( t\right)\}_{t\geq 0}$ on $X$. Set
$D^{+}(A_{0})=D(A_{0})\cap X^{+}$. Then $D(A_{0})$ is dense in $X$,
and $D^{+}(A_{0})$ is dense in $X^{+}$. Let $P:D(A_0)\to X$ be a
positive linear operator, namely, $Pf\in X^+$ for all $f\in
D^+(A_0)$. The next statement is an adaptation of Theorem~2.2 in
\cite{TV}.
\begin{proposition}
\label{le:substoch}
Suppose that, for any $f\in D^+(A_0)$,
\begin{equation}\label{cond:substoch}
 \int_{E}\bigl( ( A_{0}+P) f\bigr) \left( x\right) \nu \left(d
x\right) \leq 0.
\end{equation}
Then, for all $r\in[0,1)$, the operator $\bigl(A_0+rP, D(A_0)\bigr)$
is the generator of a substochastic $C_0$-semigroup in $X$.
\end{proposition}
Now we apply Proposition~\ref{le:substoch} to the operator (\ref{z1}).
\begin{lemma}\label{1lm}
Let $\theta$ and $\alpha^*$ be as in (\ref{z14}) and (\ref{z16}).
Then, for any $\alpha\leq \alpha^*$, the operator $A$ given by
(\ref{z1}) with ${\rm Dom}(A) = \mathcal{D}(A_1)$, is the generator
of a substochastic analytic semigroup $\{S(t)\}_{t\geq0}$ in
$\mathcal{G}_{\alpha}$.
\end{lemma}
\begin{proof}
We apply Proposition~\ref{le:substoch} with $E=\Gamma_0$,
$X=\mathcal{G}_\alpha$ as in (\ref{z3}), and $A_0=A_1$. For $r>0$
and $A_2$ as in (\ref{z111}), we set $P=r^{-1}A_2$. The cone
$\mathcal{G}_\alpha^{+}$ contains all positive elements of
$\mathcal{G}_\alpha$. For such $A_0$ and $P$, and for $G \in
\mathcal{G}_\alpha^{+} \cap \mathcal{D}(A_1)$, the left-hand side of
(\ref{cond:substoch}) takes the form, cf. (\ref{z7}),
\begin{eqnarray*}
&& - \int_{\Gamma_0} E(\eta) G(\eta) \exp(- \alpha |\eta|) \lambda (d \eta) \\[.2cm]
&& + r^{-1} \int_{\Gamma_0} \int_{\mathbb{R}^d} E^{+}(y, \eta) G(\eta\cup y) \exp(-\alpha|\eta|)dy \lambda(d\eta)\nonumber \\[.2cm]
& = & \int_{\Gamma_0} \bigl( - E(\eta) + r^{-1} e^\alpha E^{+} (\eta)\bigr) G(\eta) \exp(-\alpha|\eta|) \lambda (d\eta).\nonumber
\end{eqnarray*}
For a fixed $\alpha \leq \alpha^*$, pick $r\in (0,1)$ such that $r^{-1} e^\alpha \theta <1$, cf. (\ref{z16}). Then, for such $\alpha$ and $r$, we have
\begin{equation*}
 \int_{\Gamma_0} \bigl( - E(\eta) + r^{-1} e^\alpha E^{+} (\eta)\bigr) G(\eta) \exp(-\alpha|\eta|) \lambda (d\eta) \leq 0,
\end{equation*}
which holds in view of (\ref{z14}) and (\ref{15A}), (\ref{16A}).
By (\ref{z7}) and (\ref{z14}), we have
\begin{equation}
 \label{z19}
\|A_2 G\|_\alpha \leq e^\alpha\theta \|A_1 G\|_\alpha.
\end{equation}
This means that $r^{-1}A_2:\mathcal{D}(A_1)\to\mathcal{G}_\alpha$. Since $r^{-1}A_2$ is a positive operator, cf.~(\ref{z111}), by Proposition~\ref{le:substoch} we have that $A=A_1 + A_2 = A_1 + r (r^{-1} A_2)$ generates a substochastic semigroup $\{S(t)\}_{t\geq 0}$. Let us prove that this semigroup is analytic.

For an appropriate $\zeta\in \mathbb{C}$ and the resolvents of $A$ and $A_1$, we have
\begin{equation}
 \label{z20}
R(\zeta, A) = R(\zeta, A_1) \sum_{n=0}^\infty Q^n(\zeta), \qquad Q(\zeta) := A_2 R(\zeta, A_1).
\end{equation}
For $G\in \mathcal{G}_\alpha$,
\begin{equation*}
(Q(\zeta) G ) (\eta) = \int_{\mathbb{R}^d} \frac{E^{+} (y, \eta)}{\zeta + E(\eta\cup y)} G(\eta\cup y) d y.
\end{equation*}
Thus, for ${\rm Re}\, \zeta =: \sigma >0$, by (\ref{12A}) we obtain
\begin{eqnarray*}
 \|(Q(\zeta) G )\|_\alpha & \leq & \int_{\Gamma_0} \int_{\mathbb{R}^d} \frac{E^{+} (y, \eta)}{\sigma + E(\eta\cup y)}
| G(\eta\cup y)| \exp(- \alpha |\eta|) d y \lambda (d \eta) \\[.2cm]
& = & \int_{\Gamma_0} \frac{| G(\eta)|}{\sigma + E(\eta)}
 \exp(- \alpha |\eta|+ \alpha) \left(\sum_{x\in \eta} E^{+} (x, \eta\setminus x)\right) \lambda (d \eta) \nonumber \\[.2cm]
& \leq & \theta e^\alpha \int_{\Gamma_0} \frac{| G(\eta)|}{\sigma + E(\eta)} E(\eta)\exp(- \alpha |\eta|) \lambda (d \eta)
 \nonumber \\[.2cm]
& \leq & \theta e^\alpha \|G\|_\alpha, \nonumber
\end{eqnarray*}
where we have taken into account (\ref{15A}) and (\ref{z14}). Note that the latter estimate is uniform in $\zeta$. We use it in (\ref{z20}) and obtain
\begin{equation}
 \label{z22}
\|R(\zeta , A)\| \leq \frac{1}{1 - \theta e^\alpha} \|R(\zeta , A_1)\|.
\end{equation}
For $\zeta = \sigma + {\rm i} \tau$ with $\sigma >0$ and $\tau \neq 0$,
readily $\|R(\zeta , A_1)G\|_\alpha \leq |\tau|^{-1} \|G\|_\alpha$. Employing this estimate and (\ref{z22}) we get
\[
\|R(\sigma + {\rm i}\tau , A)\| \leq \frac{1}{|\tau|(1 - \theta e^\alpha)}.
\]
Then we apply Theorem 4.6 of \cite{EN} page 101, and obtain the analyticity of $\{S(t)\}_{t\geq 0}$, which completes the proof.
\end{proof}

As a corollary, we immediately get the solution of the problem (\ref{z18}) in the form
\begin{equation*}
G^{(0)}_t = S(t) G_0, \qquad t\geq 0,
\end{equation*}
from which we see that $G^{(0)}_t \in \mathcal{G}_{\alpha_*}$ since $G_0 \in \mathcal{G}_{\alpha_*}$, and the map $t\mapsto G^{(0)}_t$ is continuously differentiable on $(0, +\infty)$.
\begin{proof}[Proof of Theorem~\ref{1tm}] Let $\alpha_*$ and $\alpha^*$ be as in the statement of the theorem, and then let $T_*$ be as in (\ref{z15}). Now we fix $n\in \mathbb{N}$ in (\ref{z17}) and take $\alpha \in (\alpha_*, \alpha^*)$. Set
\begin{equation}
 \label{z24}
T= \frac{\alpha - \alpha_*}{ \alpha^* - \alpha_*} T_*, \qquad \  \epsilon = (\alpha - \alpha_*)/n,
\qquad \alpha_l = \alpha_* + l \epsilon,
\ \ l= 0, \dots , n.
\end{equation}
By (\ref{z13}), we have
\begin{equation}
 \label{z25}
\|B\|_{\alpha_{l-1}\alpha_l} \leq \frac{n}{e T},\qquad l = 1, \dots , n,
\end{equation}
where $\|B\|_{\alpha_{l-1}\alpha_l}$ stands for the norm in the space of all bounded linear operators from $\mathcal{G}_{\alpha_{l-1}}$ to
$\mathcal{G}_{\alpha_{l}}$. For $l=1, \dots , n$, let us consider the Cauchy problem (\ref{z17}) in $\mathcal{G}_{\alpha_l}$, i.e.,
\begin{equation}
 \label{z26}
\frac{d}{dt} G^{(l)}_t = AG^{(l)}_t +B G^{(l-1)}_t, \qquad G^{(l)}_t|_{t=0} = G_0.
\end{equation}
Assume that $G^{(l-1)}_t \in \mathcal{G}_{\alpha_{l-1}}$ is continuously differentiable on $(0,+\infty)$. Note that this assumption holds true for $l=1$. Then, by (\ref{z25}), $B G^{(l-1)}_t \in \mathcal{G}_{\alpha_{l}}$ is continuously differentiable, and hence locally H\"older continuous on $(0,+\infty)$ and integrable on $[0, \tau]$, for any $\tau>0$. By our Lemma~\ref{1lm} and Corollary 3.3, page 113 in \cite{Pazy}, this yields that the problem (\ref{z26}) on the time interval $[0, +\infty)$ has a unique classical solution in $\mathcal{G}_{\alpha_l}$, given by the formula
\begin{equation}
 \label{z27}
G^{(l)}_t = S(t) G_0 + \int_0^t S(t-s) B G^{(l-1)}_s d s.
\end{equation}
By the very definition of a classical solution, it is continuously differentiable on $(0,+\infty)$, and hence we can proceed until $l=n$. Reiterating (\ref{z27}) we obtain
\begin{eqnarray}
 \label{z28}
G^{(n)}_t & = & S(t) G_0 + \sum_{l=1}^n \int_0^t \int_0^{t_1} \cdots \int_0^{t_{l-1}} S(t-t_1) B \\[.2cm]
&& \qquad\qquad\times  S(t_1-t_2)B \cdots S(t_{l-1}-t_l)B S(t_l) G_0
dt_1 \cdots dt_l. \nonumber
\end{eqnarray}
Note that $G^{(n)}_t \in \mathcal{G}_{\alpha}$ and $\alpha = \alpha_n$, $\alpha_*=\alpha_0$, see (\ref{z24}).
From the latter representation we readily obtain
\begin{eqnarray}
\label{z29}
 \|G^{(n)}_t - G^{(n-1)}_t\|_{\alpha_n} & \leq &
 \|G_0\|_{\alpha_*} \int_0^t \int_0^{t_1} \cdots \int_0^{t_{n-1}} \|B\|_{\alpha_0\alpha_1}\\[.2cm]
&&\times \|B\|_{\alpha_1\alpha_2} \cdots \|B\|_{\alpha_{n-1}\alpha_n} dt_1 \cdots dt_n \nonumber \\[.2cm]
& \leq & \frac{1}{n!} \left( \frac{n}{e}\right)^n \left(\frac{t}{T}\right)^n \|G_0\|_{\alpha_*}, \nonumber
\end{eqnarray}
where we have used (\ref{z25}) and the fact that $\|S(t)\| \leq 1$ for all $t\geq0$, see Lemma~\ref{1lm}. For any $t\in [0,T)$, the right-hand side of the latter estimate is summable in $n$; hence, $\{G^{(n)}_t\}_{n\in \mathbb{N}_0}$ is a Cauchy sequence in $\mathcal{G}_{\alpha}$. Its limit $G_t$ is an analytic function of $t$ on the disc $\{t\in \mathbb{C}: |t|< T\}$, and thus is continuously differentiable there. Since $G_t\in \mathcal{G}_{\alpha}$, we have
\[
 G_t \in {\rm Dom} (\widehat{L}) \subset \mathcal{G}_{\alpha^*},
\]
see (\ref{z11}), (\ref{z12}), and also (\ref{z19}). For any $\alpha' \in (\alpha, \alpha^*]$, by (\ref{z17}) the sequence $\{d G^{(n)}_t/dt\}_{n\in \mathbb{N}_0}$ converges in $\mathcal{G}_{\alpha'}$ to $\widehat{L}G_t$, where we consider $\widehat{L}$ as a bounded operator from $\mathcal{G}_{\alpha}$ to $\mathcal{G}_{\alpha'}$, cf. Lemma~\ref{LLpn}. Thus, $G_t$ is a classical solution of (\ref{21}).

Now we prove the stated uniqueness. Let $\widetilde{G}_t\in \mathcal{G}_{\alpha^*}$ be another solution of the problem (\ref{21}) with the same initial $G_0\in \mathcal{G}_{\alpha_*}$, which has the properties stated in the theorem, i.e., which exists for every $\alpha^* > \alpha_*$ on the corresponding time interval. Then, as above, one can show that $\widetilde{G}_t$ is analytic at $t=0$, and
\[
 \frac{d^n}{dt^n}\widetilde{G}_t|_{t=0} = \frac{d^n}{dt^n}{G}_t|_{t=0} = \widehat{L}^n G_0 \in \mathcal{G}_{\alpha^*},
\]
where $\widehat{L}^n$ is considered as a bounded operator from $\mathcal{G}_{\alpha_*}$ to $\mathcal{G}_{\alpha^*}$, the norm of which can be estimated by (\ref{z11Ab}). Since the above holds for all $n\in \mathbb{N}$, both solutions $\widetilde{G}_t$ and $G_t$ coincide.
\end{proof}

\begin{remark}
From the proof given above one readily concludes that the evolution described by the problem (\ref{21})
takes place in the scale of spaces $\{\mathcal{G}_{\alpha}\}_{\alpha \in [\alpha_*, \alpha^*]}$ in the following sense.
For every $t \in (0, T_+)$, there exists $\alpha_t \in (\alpha_*, \alpha^*)$ such that the solution $G_t$ lies in $\mathcal{G}_{\alpha_t}\subset \mathcal{G}_{\alpha^*}$.
\end{remark}

\section{The Evolution of Correlation Functions}
\label{Sec4}
\subsection{Setting}
For the Banach space $\mathcal{G}_\alpha$ (\ref{z3}), the dual space with respect to (\ref{19A}) is
\begin{equation}
 \label{z300}
\mathcal{K}_\alpha = \{k:\Gamma_0\to \mathbb{R}:
\|k\|_\alpha < \infty\},
\end{equation}
with the norm, see (\ref{19A}),
\begin{equation}
 \label{z30}
\|k\|_\alpha = \esssup_{\eta \in \Gamma_0} |k(\eta)|\exp(\alpha|\eta|).
\end{equation}
For $\alpha'' < \alpha'$,
we have $\|k\|_{\alpha''} \leq \|k\|_{\alpha'}$ ; and hence, cf. (\ref{z5}),
\begin{equation}
 \label{z31}
\mathcal{K}_{\alpha'} \hookrightarrow \mathcal{K}_{\alpha''}, \qquad \ {\rm for} \ \alpha'' < \alpha'.
\end{equation}
The above embedding is continuous but not dense. In the sequel, we
always suppose that (\ref{z14}) and (\ref{z16}) hold, and tacitly
assume that $\alpha < \alpha^*$ for each $\alpha$ we are dealing
with. Let $A$ be defined on $\mathcal{G}_\alpha$ by (\ref{z1}), and
let $A^*$ be its adjoint in $\mathcal{K}_\alpha$ with
\begin{equation*}
{\rm Dom}(A^*) =\bigl\{k\in \mathcal{K}_\alpha: \exists
\tilde{k}\in \mathcal{K}_\alpha \ \ \forall G\in \mathcal{D}(A) \ \ \langle\! \langle AG, k\rangle \!\rangle =
\langle\! \langle G, \tilde{k}\rangle\!\rangle\bigr\}.
\end{equation*}
Then, for $A^*$ and $A^\Delta$ defined by (\ref{22A}),
we have
\begin{equation*}
A^* k = A^\Delta k = A^\Delta_1 k + A^\Delta_2 k,
\end{equation*}
which holds for all
 $k\in \mathcal{K}_\alpha$ such
that both $A^\Delta_1$ and $A^\Delta_2$ map into $\mathcal{K}_\alpha$. Let
$\mathcal{Q}_\alpha$ stand for the closure of ${\rm Dom}(A^*)$ in $\|\cdot\|_\alpha$.
Then, cf. (\ref{z12}),
\begin{equation}
 \label{z33}
\mathcal{Q}_\alpha:= \overline{{\rm Dom}(A^*)}\supset {\rm Dom}(A^*) \supset
 \mathcal{K}_{\alpha'}, \qquad {\rm for} \ {\rm any} \ \alpha'>\alpha.
\end{equation}
The latter inclusion in (\ref{z33}) follows from (\ref{z11a}) and the next obvious estimates:
\begin{eqnarray}
\label{AC}
 \|A^\Delta_1 k\|_\alpha &\leq & \|k\|_{\alpha'}
\esssup_{\eta \in \Gamma_0} E(\eta)\exp\left(-(\alpha'- \alpha)|\eta| \right), \\[.2cm]
 \|A^\Delta_2 k\|_\alpha &\leq & \esssup_{\eta \in \Gamma_0}
e^{\alpha |\eta|}\sum_{x\in \eta}E^{+}(x,\eta\setminus x) |k(\eta\setminus x)|\nonumber\\[.2cm]
&\leq & \|k\|_{\alpha'}e^{\alpha'}\esssup_{\eta \in \Gamma_0} E^{+}(\eta)\exp\left(-(\alpha'- \alpha)|\eta| \right).
\nonumber
\end{eqnarray}
Noteworthy, $\mathcal{Q}_\alpha$ is a proper subspace of $\mathcal{K}_\alpha$.

Let $\{S(t)\}_{t\geq 0}$ be the semigroup as in Lemma~\ref{1lm}. For every $t>0$, let $ S^{\odot}(t) $
denote the restriction of $S(t)^*$ to $\mathcal{Q}_\alpha$. Since
$\{S(t)\}_{t\geq 0}$ is the semigroup of contractions, for $k\in \mathcal{Q}_\alpha$ we have that, for all $t\geq 0$,
\begin{equation}
 \label{ACa}
\|S^{\odot}(t)k\|_\alpha = \|S^{*}(t)k\|_\alpha \leq \|k\|_\alpha.
\end{equation}
For any $\alpha' >\alpha$ and $t\geq 0$, in view of (\ref{z31}) we can consider $S^{\odot}(t)$ as a bounded operator
from $\mathcal{K}_{\alpha'}$ to
$\mathcal{K}_{\alpha}$, for which by (\ref{ACa}) we have
\begin{equation}
 \label{ACb}
\|S^{\odot}(t)\|_{\alpha'\alpha} \leq 1, \qquad t\geq 0.
\end{equation}
\begin{proposition}
 \label{1pn}
For every $\alpha' >\alpha$ and any $k\in \mathcal{K}_{\alpha'}$, the map
\begin{equation*}
[0, +\infty) \ni t \mapsto S^{\odot}(t) k \in \mathcal{K}_\alpha
\end{equation*}
is continuous.
\end{proposition}
\begin{proof}
By Theorem 10.4, page 39 in \cite{Pazy}, the collection
$\{S^{\odot}(t)\}_{t\geq 0}$ constitutes a $C_0$-semigroup on
$\mathcal{Q}_\alpha$, which in view of (\ref{z33}) yields the
continuity in question.
\end{proof}
By Theorem 10.4, page 39 in \cite{Pazy}, the generator of the semigroup $\{S^{\odot}(t)\}_{t\geq 0}$ is the part of $A^*$ in $\mathcal{Q}_\alpha$,
which we denote by $A^{\odot}$. Hence, by Definition 10.3, page 39 in \cite{Pazy},
$A^{\odot}$ is the restriction of $A^*$ to the set
\begin{equation}
 \label{z34a}
{\rm Dom} (A^{\odot}):= \{ k\in {\rm Dom}(A^*): A^*k \in \mathcal{Q}_\alpha\}.
\end{equation}
For $\alpha'>\alpha$, we take $\alpha''\in (\alpha, \alpha')$ and obtain by (\ref{AC}) that
\begin{equation*}
A^*: \mathcal{K}_{\alpha'} \to \mathcal{K}_{\alpha''}.
\end{equation*}
Hence, for any $\alpha'>\alpha$,
\begin{equation}
 \label{AE}
{\rm Dom} (A^{\odot}) \supset \mathcal{K}_{\alpha'}.
\end{equation}
We recall that each $k$ may be identified with a sequence $\{k^{(n)}\}_{n\in \mathbb{N}_0}$
of symmetric $k^{(n)} \in L^\infty\left((\mathbb{R}^{d})^n\right)$, $k^{(0)}\in \mathbb{R}$.
Put $q^{(n)}= \|k^{(n)}\|_{L^\infty(\mathbb{R}^{nd})}$, $q^{(0)} = |k^{(0)}|$. Then (\ref{z30}) can be rewritten in the form
\begin{equation*}
 \|k\|_\alpha = \sup_{n\in \mathbb{N}_0} q^{(n)} e^{n \alpha}.
\end{equation*}
Set, cf. (\ref{23A}),
\begin{equation*}
\mathcal{D}(B^\Delta) = \{ k \in \mathcal{K}_\alpha : \sup_{n\in \mathbb{N}_0}n q^{(n)} e^{\alpha n}< \infty\}.
\end{equation*}
Then, see (\ref{13A}),
\begin{eqnarray*}
\|B^\Delta_1 k\|_\alpha & \leq & \sup_{n\in \mathbb{N}_0} q^{(n+1)} e^{\alpha n} \sup_{\eta \in \Gamma^{(n)}} \int_{\mathbb{R}^d} E^{-}(y,\eta) dy\\[.2cm]
& = & \langle{a}^{-} \rangle  \sup_{n\in \mathbb{N}_0}n q^{(n+1)} e^{\alpha n}
\leq e^{-\alpha} \langle{a}^{-} \rangle \sup_{n\in \mathbb{N}_0}n q^{(n)} e^{\alpha n}.
\end{eqnarray*}
$\|B^\Delta_2 k\|_\alpha$ can be estimated in the same way, which then yields
\begin{equation}
 \label{z33b}
\|B^\Delta k\|_\alpha \leq \Bigl(\langle{a}^{+}\rangle + \langle{a}^{-} \rangle e^{-\alpha}\Bigr) \sup_{n\in \mathbb{N}_0}n q^{(n)} e^{\alpha n}.
\end{equation}
Hence, $B^\Delta$ maps $\mathcal{D}(B^\Delta)$ into $\mathcal{K}_\alpha$. Let $(B^*, {\rm Dom}(B^*))$ be the adjoint operator to $(B, {\rm Dom}(B))$.
Then
$B^* k = B^\Delta k$ for $k\in \mathcal{D}(B^\Delta)$, and
\begin{equation}
 \label{z33c}
{\rm Dom}(B^*) \supset \mathcal{D}(B^\Delta)\supset
 \mathcal{K}_{\alpha'}, \quad {\rm for} \ {\rm any} \ \alpha'>\alpha.
\end{equation}
The latter inclusion follows from the estimate, cf. (\ref{z13}) and (\ref{z25}),
\begin{equation}
 \label{z33d}
 \|B^*\|_{\alpha'\alpha} \leq \frac{\langle{a}^{+} \rangle + \langle{a}^{-} \rangle e^{-\alpha}}{e(\alpha'- \alpha)},
\end{equation}
which can easily be obtained from (\ref{z33b}). Now we can define $L^\Delta$ as an operator in $\mathcal{K}_\alpha$. Namely, we set
\begin{eqnarray}
 \label{z33e}
L^\Delta & = & A^{\odot} + B^\Delta, \\[.2cm]
{\rm Dom} (L^\Delta) & = & {\rm Dom}(A^{\odot}) \cap \mathcal{D} (B^{\Delta}). \nonumber
\end{eqnarray}
By (\ref{AE}) and (\ref{z33c}), for any $\alpha'>\alpha$ we have
\begin{equation*}
{\rm Dom} (L^\Delta) \supset \mathcal{K}_{\alpha'}.
\end{equation*}

\subsection{The statement}

\begin{theorem}
 \label{2tm}
Let $\theta$, $\alpha_*$, $\alpha^*$, and $T_*$ be as in Theorem~\ref{1tm}. Then for every $k_0\in \mathcal{K}_{\alpha^*}$,
the problem (\ref{20A}) has a unique classical solution in $\mathcal{K}_{\alpha_*}$ on the time interval $[0,T_*)$.
\end{theorem}
\begin{proof}
Let $k_0\in \mathcal{K}_{\alpha^*}$, $\alpha \in(\alpha_*, \alpha^*)$, $n\in \mathbb{N}$, and $l=1, \dots , n$ be fixed. Consider
\begin{equation}
 \label{z37}
 K_l (t, t_1, \dots , t_l):= S^{\odot}(t-t_1) B^*S^{\odot}(t_1-t_2)B^* \cdots S^{\odot}(t_{l-1}-t_l)B^* S^{\odot}(t_l) k_0,
\end{equation}
where the arguments $(t,t_1, \dots, t_l)$ belong to the set
\begin{equation}
 \label{cone}
\mathcal{T}_l :=\{ (t,t_1, \dots, t_l) : 0\leq t_l \leq \cdots \leq t_1 \leq t\}.
\end{equation}
In (\ref{z37}), we mean that the operators act in the following spaces, cf. (\ref{ACb}),
\begin{equation*}
 S^{\odot}(t_l) : \mathcal{K}_{\alpha_0} \to \mathcal{K}_{\alpha_1}, \quad S^{\odot}(t_{l-s}-t_{l-s+1}) : \mathcal{K}_{\alpha_{2s}}\to \mathcal{K}_{\alpha_{2s+1}}, \quad s = 1, \dots, l,
\end{equation*}
and, cf. (\ref{z33d}),
\begin{equation}
 \label{z39}
 B^* : \mathcal{K}_{\alpha_{2s-1}}\to \mathcal{K}_{\alpha_{2s}}, \quad s = 1, \dots, l.
\end{equation}
Here, for
a positive $\delta< \alpha^* - \alpha$, we set
\begin{eqnarray}
 \label{z40}
 \alpha_{2s} & = & \alpha^*- \frac{s}{l+1}\delta - s \epsilon, \qquad \epsilon = (\alpha^* - \alpha - \delta)/l, \\[.2cm]
 \alpha_{2s+1} & = & \alpha^*- \frac{s+1}{l+1}\delta - s \epsilon, \qquad s= 0,1, \dots, l. \nonumber
\end{eqnarray}
Note that $\alpha_0 = \alpha^*$ and $\alpha_{2l+1} = \alpha$, and hence $K_l(t,t_1, \dots, t_l)
 \in \mathcal{K}_\alpha$.
In view of Proposition~\ref{1pn} and (\ref{z33d}), $K_l$ is a continuous function of each of its variables on (\ref{cone}).
Furthermore, it is differentiable in $t\in (0,+\infty)$ in every $\mathcal{K}_{\alpha'}$, $\alpha' \in (\alpha_*, \alpha)$,
and the following holds, cf. (\ref{z34a}) and (\ref{AE}),
\begin{equation*}
\frac{d}{dt}K_l(t,t_1, \dots, t_l) = A^{\odot}K_l(t,t_1, \dots, t_l).
\end{equation*}
Now we set
\begin{equation}
 \label{z41}
k^{(n)}_t = S^{\odot}(t) k_0 + \sum_{l=1}^n \int_0^t \int_0^{t_1} \cdots \int_0^{t_{l-1}}K_l (t, t_1, \dots , t_l) dt_1 \cdots dt_l.
\end{equation}
For
\begin{equation}
 \label{z42}
 T_\delta := \frac{\alpha^* - \alpha - \delta}{\alpha^* - \alpha_*}T_*,
\end{equation}
the function $[0, T_\delta)\ni t \mapsto k^{(n)}_t \in \mathcal{K}_\alpha $ is continuous, whereas
$(0, T_\delta)\ni t \mapsto k^{(n)}_t \in \mathcal{K}_{\alpha'} $ is differentiable, and the following holds,
cf. (\ref{z26}),
\begin{equation}
 \label{z43}
\frac{d}{dt} k^{(n)}_t = A^{\odot} k^{(n)}_t + B^* k^{(n-1)}_t, \qquad k^{(n)}_t|_{t=0} = k_0.
\end{equation}
For $T<T_*$, let us show that there exists $\alpha \in (\alpha_*, \alpha^*)$ such that
the sequence $\{k^{(n)}_t\}_{n\in \mathbb{N}}$ converges in $\mathcal{K}_\alpha$ uniformly on
$[0,T]$. For this $T$, we pick $\alpha \in (\alpha_*, \alpha^*)$ and a positive $\delta < \alpha^* - \alpha$ such that
also $T< T_\delta$, see (\ref{z42}). As in (\ref{z29}), for $t\in [0,T]$ we get
\begin{eqnarray*}
\|k^{(n)}_t - k^{(n-1)}_t\|_\alpha & \leq & \int_0^t \int_0^{t_1} \cdots \int_0^{t_{n-1}}\|K_n (t, t_1, \dots , t_n)\|_{\alpha} dt_1 \cdots dt_n\\[.2cm]
& \leq & \frac{T^n}{n!}\|k_0\|_{\alpha^*} \prod_{s=1}^n \|B^* \|_{\alpha_{2s-1}\alpha_{2s}},
\end{eqnarray*}
where we have taken into account (\ref{ACb}) and (\ref{z39}), (\ref{z40}) with $l=n$. Then by means of (\ref{z33d}) we obtain
\begin{eqnarray}
\label{z44}
\|k^{(n)}_t - k^{(n-1)}_t\|_\alpha & \leq & \frac{T^n}{n!}\left(\frac{n}{e} \right)^n
 \left( \frac{\alpha^* - \alpha_*}{(\alpha^* - \alpha - \delta)T_*} \right)^n \|k_0\|_{\alpha^*}\\[.2cm]
& = &\frac{1}{n!}\left(\frac{n}{e} \right)^n \left(\frac{T}{T_\delta} \right)^n \|k_0\|_{\alpha^*}, \nonumber
\end{eqnarray}
which certainly yields the convergence to be proven. Now we take $\alpha'\in [\alpha_*, \alpha)$ and obtain the convergence of both sides of
(\ref{z43}) in $\mathcal{K}_{\alpha'}$ where both operators are considered as bounded operators acting from
$\mathcal{K}_{\alpha}$ to $\mathcal{K}_{\alpha'}$, see (\ref{ACb}) and (\ref{z33d}). This yields that the limit $k_t\in \mathcal{K}_{\alpha_*}$ of the sequence $\{k^{(n)}_t\}_{n\in \mathbb{N}}$ solves (\ref{20A}) with $L^\Delta$ given by (\ref{z33e}).
\end{proof}
\begin{remark}
 \label{z2rk}
From the proof given above one concludes that the evolution described by the problem (\ref{20A})
takes place in the scale of spaces $\{\mathcal{K}_{\alpha}\}_{\alpha \in [\alpha_*, \alpha^*]}$ in the sense that,
for every $t \in (0, T_*)$, there exists $\alpha_t \in (\alpha_*, \alpha^*)$ such that the solution $k_t$ lies in $\mathcal{K}_{\alpha_t}\subset \mathcal{K}_{\alpha_*}$.
\end{remark}

\subsection{The dual evolutions}
Recall that the duality between correlation functions and quasi-observables is established by the relation (\ref{19A}).
\begin{definition}
Let $\alpha_*$, $\alpha^*$, $T_*$ be as in Theorem~\ref{1tm}, and for $G_0 \in \mathcal{G}_{\alpha_*}$, let $G_t$ be the solution
of the problem (\ref{21}). For a given $k_0 \in \mathcal{K}_{\alpha^*}$, the dual evolution $k_0 \mapsto k_t^D$ is the weak$^*$-continuous map
$[0,T_*) \ni t \mapsto k_t^D \in \mathcal{K}_{\alpha_*}$ such that, for every $t\in[0,T_*)$, the following holds
\begin{equation}
 \label{z46}
 \langle \! \langle G_t, k_0 \rangle \! \rangle = \langle \! \langle G_0, k_t^D \rangle \! \rangle.
\end{equation}
Likewise, for $k_0 \in \mathcal{K}_{\alpha^*}$, let $k_t$ be the solution
of the problem (\ref{20A}), see Theorem~\ref{2tm}. For a given $G_0 \in \mathcal{G}_{\alpha_*}$, the dual evolution
$G_0 \mapsto G_t^D$ is the weak-continuous map
$[0,T_*) \ni t \mapsto G_t^D \in \mathcal{G}_{\alpha^*}$ such that, for every $t\in[0,T_*)$, the following holds
\begin{equation}
 \label{z47}
 \langle \! \langle G_0, k_t \rangle \! \rangle = \langle \!\langle G^D_t, k_0 \rangle \!\rangle.
\end{equation}
\end{definition}
Note that the solution of (\ref{20A}) need not coincide with $k_t^D$, and similarly, the solution
of (\ref{21}) need not be the same as $G_t^D$. It is even not obvious whether such dual evolutions exist since the topological dual
to $\mathcal{K}_{\alpha}$ is not $\mathcal{G}_\alpha$.
\begin{theorem}
 \label{3tm}
For any $G_0 \in \mathcal{G}_{\alpha_*}$ and any $k_0 \in \mathcal{K}_{\alpha^*}$, the dual evolutions $k_0 \mapsto k^D_t$ and $G_0 \mapsto G^D_t$ exist and are norm-continuous.
\end{theorem}
\begin{proof}
First we prove the existence of $k^D_t$. For a given $k_0 \in \mathcal{K}_{\alpha^*}$ and a fixed $n\in \mathbb{N}$, let $\alpha$, $\delta$, and $l$ be as in the proof of Theorem~\ref{2tm}. Set
 \begin{eqnarray}
 \label{z36}
K^D_l(t, t_1 , \dots , t_{l}) & := & S^{\odot}(t_{l}) B^* S^{\odot}(t_{l-1} - t_{l}) B^* \cdots S^{\odot}(t_{1} - t_{2}) B^* \\[.2cm]
&& \times S^{\odot}(t - t_{1}) k_0, \nonumber
\end{eqnarray}
where the above operators act in the following spaces
\begin{eqnarray*}
& & S^{\odot}(t_s - t_{s+1}): \mathcal{K}_{\alpha_{2s}} \to \mathcal{K}_{\alpha_{2s+1}}, \quad \ s=0,1 , \dots , l-1, \\[.2cm] & & S^{\odot}(t - t_{1}): \mathcal{K}_{\alpha_0} \to \mathcal{K}_{\alpha_1}, \qquad S^{\odot}(t_l):\mathcal{K}_{\alpha_{2l}} \to \mathcal{K}_{\alpha_{2l+1}} ,
\end{eqnarray*}
and $B^*$ act as in (\ref{z39}). The numbers $\alpha_s$ are given by (\ref{z40}). Then we set, cf. (\ref{z41}),
\begin{equation}
 \label{z49}
k^{D,n}_t = S^{\odot}(t) k_0 + \sum_{l=1}^n \int_0^t \int_0^{t_1} \cdots \int_0^{t_{l-1}}K^D_l (t, t_1, \dots , t_l) dt_1 \cdots dt_l.
\end{equation}
Exactly as in the proof of Theorem~\ref{2tm} we obtain, cf. (\ref{z44}),
\[
\|k^{D,n}_t - k^{D,n-1}_t\|_\alpha \leq \frac{1}{n!}\left(\frac{n}{e} \right)^n \left(\frac{T}{T_\delta} \right)^n \|k_0\|_{\alpha^*}
\]
which yields that the sequence $\{k^{D,n}_t\}_{n\in \mathbb{N}}$ converges in
$\mathcal{K}_\alpha$ uniformly on $[0,T]$. Hence its limit, which we denote by
$k^{D}_t$, is a norm-continuous function from $[0,T_*)$ to $\mathcal{K}_\alpha$, and $\mathcal{K}_\alpha\hookrightarrow \mathcal{K}_{\alpha_*}$. Note that
$k^{D}_t\in \mathcal{K}_{\alpha_t}$ where $\alpha_t\in (\alpha_*, \alpha^*)$, cf. Remark~\ref{z2rk}.

For every $G\in \mathcal{G}_{\alpha_*}$, the map $\mathcal{K}_{\alpha_*} \ni k \mapsto \langle \! \langle G, k \rangle \! \rangle \in \mathbb{R}$
is continuous. Since each $K_l^D$ in (\ref{z49}) is in $\mathcal{K}_{\alpha_*}$, we have, cf. (\ref{z36}),
\begin{eqnarray}
 \label{z50}
& & \Bigl\langle \!\! \Bigl\langle G_0, \int_0^t \int_0^{t_1} \cdots \int_0^{t_{l-1}}K^D_l (t, t_1, \dots , t_l) dt_1 \cdots dt_l \Bigr\rangle \!\! \Bigr\rangle \\[.2cm]
& & \qquad = \int_0^t \int_0^{t_1} \cdots \int_0^{t_{l-1}}
\langle \! \langle G_0, K^D_l (t, t_1, \dots , t_l) \rangle \! \rangle dt_1 \cdots dt_l \nonumber \\[.2cm]
& & \qquad = \int_0^t \int_0^{t_1} \cdots \int_0^{t_{l-1}} \langle \! \langle S(t-t_1) B S(t_1-t_2) B \nonumber \\[.2cm]
& & \qquad \qquad\qquad\times \cdots \times S(t_{l-1} - t_l)B S(t_l) G_0,
k_0 \rangle \! \rangle dt_1 \cdots dt_l \nonumber.
\end{eqnarray}
Thereafter, by (\ref{z49}) we obtain, cf. (\ref{z28}),
\begin{equation*}
\langle \! \langle G_0, k^{D,n}_t \rangle \! \rangle = \langle \! \langle G_t^{(n)}, k_0 \rangle \! \rangle,
\end{equation*}
which holds for all $t\in [0,T_*)$ and $n\in \mathbb{N}$. Passing here to the limit $n\to \infty$ and taking into account
the norm convergences $G_t^{(n)} \to G_t$, see Theorem~\ref{1tm}, and
$k^{D,n}_t \to k^{D}_t$ established above, we arrive at (\ref{z46}).

To prove (\ref{z47}), for $t\in [0,T_*)$ and $n\in \mathbb{N}$, we consider, cf. (\ref{z36}),
\begin{eqnarray*}
G^{D,n}_t & = & S(t) G_0 + \sum_{l=1}^n \int_0^t \int_0^{t_1} \cdots \int_0^{t_{l-1}} S(t_l) B \\[.2cm]
&& \qquad\qquad\times S(t_{l-1} - t_l) B \cdots S(t_1 - t_2)B S(t-t_1) G_0 d t_1 \cdots dt_l.
\end{eqnarray*}
As in the proof of Theorem~\ref{1tm}, we show that the sequence of $G^{D,n}_t$, $n\in \mathbb{N}$,
converges in $\mathcal{G}_{\alpha^*}$, uniformly on compact subsets of
$[0,T_*)$. Let $G^{D}_t$ be its limit. By the very construction, and due to the
possibility of interchanging the integrations as in (\ref{z50}), we get
\[
\langle \! \langle G^{D,n}_t, k_0 \rangle\! \rangle = \langle \! \langle G_0, k_t^{(n)} \rangle \! \rangle,
\]
where $k_t^{(n)}$ is the same as in (\ref{z41}). Passing here to the limit $n\to \infty$ we arrive at (\ref{z47}).
\end{proof}
\begin{remark}
 \label{Grk}
As in Remark~\ref{z2rk}, from the above proof we conclude that, for each $t \in (0, T_*)$, there exists $\alpha_t \in (\alpha_*, \alpha^*)$
such that $G^D_t \in \mathcal{G}_{\alpha_t}\subset \mathcal{G}_{\alpha^*}$.
\end{remark}

\section{The Evolution of States}

Theorem \ref{2tm} does not ensure that the solutions $k_t$ are correlation functions. Below we prove this holds under the condition (\ref{yy12}). Recall that we also assume (\ref{z14}).

\subsection{The evolution of local densities}

Let a measure $\mu \in \mathcal{M}^1_{\rm fm} (\Gamma)$ be locally absolutely continuous
with respect to the Poisson measure $\pi$. In that, for each $\Lambda \in \mathcal{B}_{\rm b}(\mathbb{R}^d)$, the projection
$\mu^\Lambda$ is absolutely continuous with respect to $\pi^\Lambda$, and hence to $\lambda^\Lambda$, cf. (\ref{3A}).
Consider
\begin{equation}
 \label{y1}
R^\Lambda (\eta) := \mathbb{I}_{\Gamma_\Lambda}(\eta) \frac{d\mu^\Lambda}{d\lambda^\Lambda}(\eta) , \qquad \eta \in \Gamma_0.
\end{equation}
Clearly, $R^\Lambda$ is a positive element of the Banach space $L^1 (\Gamma_0, d \lambda)$ of unit norm.
We call it a {\it local density}.
The measure $\mu$ is characterized by the correlation measure (\ref{9A}), and thus by the correlation function (\ref{9AA}) which can be written in the form (\ref{9AAghf}), cf. Proposition 4.2 in \cite{Tobi},
\begin{equation}
 \label{y2}
 k^\Lambda(\eta) = k(\eta)\mathbb{I}_{\Gamma_\Lambda}(\eta) = \int_{\Gamma_0} R^\Lambda (\eta\cup \xi) \lambda (d \xi), \qquad \eta \in \Gamma_0.
\end{equation}
Note that
$k^\Lambda =  k^\Lambda \mathbb{I}_{\Gamma_\Lambda}$.

As in \cite{Tobi1}, we say that
a probability measure $\mu$ on $\mathcal{B}(\Gamma)$ obeys Dobrushin's exponential bound with a given $\alpha >0$, if
for any $\Lambda \in \mathcal{B}_{\rm b} (\mathbb{R}^d)$, there exists $C_\Lambda>0$ such that
\begin{equation}
 \label{yy10y}
\int_{\Gamma_\Lambda} \exp( \alpha |\eta|) \mu^\Lambda (d\eta) \leq C_\Lambda.
\end{equation}
For $\alpha >0$, we set
\begin{equation}
 \label{yy11}
b_\alpha (\eta) = \exp( \alpha |\eta|), \qquad  \eta \in \Gamma_0.
\end{equation}
Clearly, if $\mu$ obeys (\ref{yy10y}) with a given $\alpha>0$, then, for all $\Lambda \in \mathcal{B}_{\rm b} (\mathbb{R}^d)$,
\begin{equation}
 \label{y1A}
 R^\Lambda \in \mathcal{R}_{\alpha}:= L^1 (\Gamma_0, b_\alpha d \lambda).
\end{equation}
In this subsection, we study the evolution of local densities in the space $\mathcal{R}_{\alpha}$.

As was mentioned above, we cannot define $L$ as given in
(\ref{R20}) on any space of functions $F:\Gamma \to \mathbb{R}$. However, it is possible to do
in the case of
bounded measurable functions $F:\Gamma_0\to \mathbb{R}$, i.e., on the space $L^\infty (\Gamma_0, d \lambda)$. Set
\begin{equation}
 \label{yy2}
 \Xi (\eta) = E(\eta) + \langle a^+ \rangle |\eta|, \qquad \eta\in\Gamma_0.
\end{equation}
Then we rewrite (\ref{R20}) in the following form
\begin{eqnarray}\notag
(LF)(\eta)&=&-\Xi(\eta)F(\eta)+\sum_{x\in \eta }\bigl( m + E^-( x,\eta\setminus x)\bigr) F\left( \eta
\setminus x\right) \notag\\&&+\int_{\mathbb{R}^{d}}E^+\left( x,\eta\right) F\left(
\eta\cup x\right) dx,\qquad\eta\in \Gamma _{0}.\label{deltainf}
\end{eqnarray}
Let $R\in L^{1}\left( \Gamma _{0},d\lambda \right) $ be such that $\Xi R\in
L^{1}\left( \Gamma _{0},d\lambda \right) $. For such $R$ and for any $F\in L^{\infty}\left( \Gamma _{0},d\lambda \right)$, by (\ref{12A}) we get
\begin{eqnarray*}
\int_{\Gamma _{0}}\left( LF\right) \left( \eta \right) R\left( \eta
\right) d\lambda \left( \eta \right)
&=&-\int_{\Gamma _{0}}\Xi\left( \eta \right) F\left( \eta \right) R\left(
\eta \right) d\lambda \left( \eta \right)  \nonumber \\
&&+\int_{\Gamma _{0}}\int_{\mathbb{R}^{d}}\bigl( m + E^-( x,\eta)\bigr) F\left(
\eta \right) R\left( \eta \cup x\right) dxd\lambda \left( \eta \right)
\nonumber \\
&&+\int_{\Gamma _{0}}\sum_{x\in \eta }E^+\left( x,\eta \setminus x\right)
F\left( \eta \right) R\left( \eta \setminus x\right) d\lambda \left(
\eta \right) .
\end{eqnarray*}
Next, we define the following operator in $L^1 (\Gamma_0, d \lambda)$
\begin{eqnarray}
 \label{y3}
( L^\dagger R )( \eta) & = & (A_0 R)(\eta) + (B R)(\eta)
 := - \Xi (\eta) R(\eta)\\[.2cm] && + \int_{\mathbb{R}^d} (m + E^{-} (y, \eta)) R( \eta\cup y) dy +
\sum_{x\in \eta} E^{+} (x, \eta \setminus x) R(\eta \setminus x) ,  \nonumber
\end{eqnarray}
with
\begin{equation}
 \label{y4}
 {\rm Dom} (L^\dagger) = \bigl\{ R\in  L^1 (\Gamma_0, d \lambda) : \Xi R \in  L^1 (\Gamma_0, d \lambda)\bigr\}.
\end{equation}
Then, for any $F\in L^\infty (\Gamma_0, d\lambda)$, we have
\begin{equation}\label{adj}
\int_{\Gamma_0} LF\cdot R\,d\lambda = \int_{\Gamma_0} F\cdot L^\dagger R \,d\lambda.
\end{equation}
\begin{lemma}\label{y3tm}
Suppose that the following condition be satisfied
\begin{equation}
 \label{yy12}
 m > \langle a^+ \rangle.
\end{equation}
Then the closure of $L^\dagger$ given in (\ref{y3}) and (\ref{y4}) is the generator of a stochastic $C_0$-semigroup $\{S^\dagger(t)\}_{t\geq 0}$ of bounded linear operators in $ L^1 (\Gamma_0, d \lambda)$, which leave invariant each $\mathcal{R}_{\alpha}$ with $\alpha\leq \log m - \log \langle a^{+} \rangle$.  Moreover, the restrictions ${S}^\dagger_\alpha (t) := S^\dagger(t)|_{\mathcal{R}_{\alpha}}$, $t\geq 0$, constitute a positive $C_0$-semigroup in $\mathcal{R}_{\alpha}$, the generator ${L}^\dagger_\alpha$ of which is the restriction of $\bigl(L^\dagger, {\rm Dom} (L^\dagger)\bigr)$.
\end{lemma}

 As in Section~\ref{Sec3}, we employ the perturbation theory for positive semigroups developed in \cite{TV}. To proceed further, we need some facts in addition to those preceding Proposition~\ref{le:substoch}. Recall that $X$ stands for $L^1(E, d\nu)$.

Let $\rho\in L^1_\mathrm{loc}(E,d\nu)$ be such that $p:=\essinf_{x\in E}\rho(x) >0$. We consider the Banach $X_{\rho}:=L^{1}\left( E,\rho \, d\nu \right) $ with norm $\left\Vert \cdot\right\Vert _{\rho }$. Clearly, $X_\rho \hookrightarrow X$, where the embedding is dense and continuous. The latter follows from the fact that $\left\Vert f\right\Vert \leq p^{-1}\left\Vert f\right\Vert _{\rho }$ for all $f\in X_{\rho }$. Next, for $X_{\rho }^{+}:=X_{\rho }\cap X^{+}$ we have that $X_{\rho }^{+}$ is dense in $X^{+}$ and $X_{\rho }=X_{\rho }^{+}-X_{\rho }^{+}$. Note that, $\left\Vert f+g\right\Vert _{\rho }=\left\Vert f\right\Vert _{\rho }+\left\Vert g\right\Vert _{\rho }$ for any $f,g\in X_{\rho }^{+}$.

Let $(A_0, D(A_0))$ be the generator of a positive $C_0$-semigroup $\{S_0(t)\}_{t\geq 0}$ of contractions on $X$. Then we set $\check{S}_{0}\left( t\right)= S_0 (t)|_{X_\rho}$, $t\geq 0$, and assume that the following holds:
\begin{enumerate}
\item [(a)] The operators $S_{0}\left( t\right)$, $t\geq 0$, leave $X_{\rho }$ invariant.
\item [(b)] $\{\check{S}_{0}\left( t\right)\}_{t\geq 0}$ is a $C_{0}$-semigroup on $X_{\rho}$.
\end{enumerate}
By Proposition II.2.3 of \cite{EN},
the generator $\check{A}_{0}$ of the semigroup $\{\check{S}_{0}\left( t\right)\}_{t\geq 0}$ is the part of $A_{0}$. Namely, $\check{A}_{0}f=A_{0}f$ on the domain
\[
D(\check{A}_{0})=\left\{ f\in D\left( A_{0}\right) \cap X_{\rho }:A_{0}f\in
X_{\rho }\right\} .
\]%
Set $D^{+}(\check{A}_{0})=D(\check{A}_{0})\cap X_{\rho }^{+}$. The next statement is an adaptation of Proposition~2.6 and Theorem~2.7 of \cite{TV}.
\begin{proposition}\label{thTV}
Let conditions (a) and (b) above hold, and $-A_0$ be a positive linear operator in $X$. Suppose also that, cf. (\ref{cond:substoch}),
\begin{equation*}
 \int_{E}\bigl( ( A_{0}+P) f\bigr) \left( x\right) \nu \left( d
x\right) = 0,
\end{equation*}
 where $P$ is such that $P\left( D(\check{A}_{0})\right) \subset X_{\rho }$. Finally, assume that there exist $c>0,\varepsilon >0$ such that, for all $f\in D^{+}(\check{A}_{0})$, the following estimate holds
\[
\int_{E}\bigl( \left( A_{0}+P\right) f\bigr) \left( x\right) \rho \left(
x\right) \nu \left(d x\right) \leq c\int_{E} f\left( x\right)
\rho \left( x\right) \nu \left( dx\right) +\varepsilon
\int_{E}\left( A_{0}f\left( x\right) \right) \nu \left(d x\right) .
\]
Then the closure $\bigl( A,D(A)\bigr) $ of the operator $%
\bigl( A_{0}+P,D\left( A_{0}\right) \bigr) $ is the generator of a
stochastic semigroup $\{S\left( t\right)\}_{ t\geq 0}$ on $X$. This semigroup
leaves the space $X_{\rho }$ invariant and induces a positive $C_{0}$%
-semigroup, $\check{S}\left( t\right) $, on $X_{\rho }$ with generator $\bigl(\check{A},D(\check{A})\bigr)$, which is the restriction of $\bigl( A_{0}+P,D\left(
A_{0}\right) \bigr) $ on $X_{\rho }$. Moreover, the operator $\bigl( A,D(A)\bigr) $ is
the closure of $\bigl(\check{A},D(\check{A})\bigr)$ in $X$.
\end{proposition}
We shall use the version of Proposition~\ref{thTV} in which $A_0$ is a
multiplication operator. Let $a:E\rightarrow \mathbb{R}_{+}$ be a measurable nonnegative function on $E$. Set
\[
\left( A_{0}f\right) \left( x\right) =-a\left( x\right) f\left( x\right), \quad x\in E, \qquad D\left( A_{0}\right) :=\left\{ f\in X:af\in X\right\} .
\]
Clearly, $-A_{0}$ is a positive operator in $X$.
Then, by, e.g., Lemma II.2.9 in \cite{EN}, $\bigl( A_{0},D\left( A_{0}\right)
\bigr) $ is the generator of the $C_{0}$-semigroup
composed by the (positive) multiplication operators
$S_{0}\left( t\right) =\exp \left\{ -ta\left( x\right) \right\}$, $t\geq 0$. For any $f\in X_{\rho }$, we have $\| S_0(t)f\|_{\rho}\leq \| f\|_{\rho}$; hence, $S_{0}\left(
t\right) $ leaves $X_{\rho }$ invariant. By,
e.g., Proposition I.4.12 and Lemma II.2.9 in \cite{EN}, the restrictions $\check{S}_{0}\left( t\right):= S_0(t)|_{X_\rho}$, $t\geq 0$, constitute a $C_{0}$-semigroup in $X_{\rho }$ with generator $\check{A}_{0}$ which acts as
$\check{A}_{0}f=A_{0}f$ on the domain
\[
D(\check{A}_{0})=\left\{ f\in X_{\rho }:af\in X_{\rho }\right\} \subset
D\left( A_{0}\right) .
\]
\begin{lemma}\label{thTV-mult}
Let $P:D(A_{0})\rightarrow X$ be a positive linear operator such that
\begin{equation*}
\int_{E}\left( Pf\right) \left( x\right) \nu \left( dx\right)
=\int_{E}a\left( x\right) f\left( x\right) \nu \left(d x\right) ,\quad ~~f\in
D^{+}(A_{0}).
\end{equation*}
Suppose also that there exist $c>0,\varepsilon >0$ such that, for all
$f\in D^{+}(\check{A}_{0})$, the following holds
\begin{align}
\int_{E}\left( Pf\right) \left( x\right) \rho \left( x\right) \nu \left(d
x\right) \leq &\int_{E}\left( c+a\left( x\right) \right) f\left( x\right)
\rho \left( x\right) \nu \left(d x\right) \notag\\&-\varepsilon \int_{E}a\left(
x\right) f\left( x\right) \nu \left(d x\right) . \label{newineq}
\end{align}%
Then the statements of Proposition~\ref{thTV} hold.
\end{lemma}
\begin{proof}
To apply Proposition~\ref{thTV} we should only show that $P\left( D(\check{A}_{0})\right) \subset X_{\rho }$.
Let us show that
it follows from \eqref{newineq}. Indeed, for $f\in D(\check{A}_{0})$, we have that both $f$ and $af$ are in $X_{\rho }$.
Set $f^{+}=\max \left\{ f;0\right\} $, $f^{-}=-\min \left\{ f;0\right\} $. Then $f^{\pm }\in X_{\rho }^{+}$ and $\left\vert f^{\pm }\right\vert
\leq \left\vert f\right\vert $, which yields $af^{\pm }\in X_{\rho }$.
Hence, $%
f^{\pm }\in D^{+}(\check{A}_{0})$, and therefore, by \eqref{newineq},
\begin{equation}\label{Bfpm-est}
\int_{E}\left( Pf^{\pm }\right) \left( x\right) \rho \left( x\right) \nu
\left(d x\right) <\infty .
\end{equation}
Since $f=f^{+}-f^{-}$ and $P$ is positive, we have by \eqref{Bfpm-est}
\begin{eqnarray*}
\left\Vert Pf\right\Vert _{X_{\rho }} &=&\int_{E}\bigl\vert (
Pf^{+}) ( x) - ( Pf^{-}) (x)
\bigr\vert \rho (x) \nu \left( d x\right) \\
&\leq &\int_{E}\Bigl( \bigl\vert ( Pf^{+}) \left( x\right)
\bigr\vert +\bigl\vert ( Pf^{-}) \left( x\right) \bigr\vert
\Bigr) \rho \left( x\right) \nu \left(d x\right) \\
&=&\int_{E}\bigl( ( Pf^{+}) \left( x\right) +( Pf^{-})
\left( x\right) \bigr) \rho \left( x\right) \nu \left(d x\right) <\infty .
\end{eqnarray*}
\end{proof}
\begin{proof}[Proof of Lemma~\ref{y3tm}]
For any $R\in {\rm Dom} (L^\dagger)$, by (\ref{12A}), we have
\begin{eqnarray*}
\int_{\Gamma _{0}}\left\vert \left( BR\right) \left( \eta \right)
\right\vert \lambda(d \eta)  &\leq &\int_{\Gamma _{0}}\int_{%
\mathbb{R}^{d}}\bigl( m + E^-( x,\eta)\bigr)\left\vert R\left( \eta \cup
x\right) \right\vert dx \lambda \left( d\eta \right) \\
&&+\int_{\Gamma _{0}}\sum_{x\in \eta }E^+\left( x,\eta\setminus x\right)
\left\vert R\left( \eta\setminus x\right) \right\vert \lambda \left(d
\eta\right) \\
&=&\int_{\Gamma _{0}}\Xi\left( \eta \right) \left\vert R\left( \eta
\right) \right\vert \lambda \left(d \eta \right) <\infty .
\end{eqnarray*}%
Then $B:{\rm Dom} (L^\dagger)\to L^{1}\left( \Gamma _{0},d \lambda
\right) $. Clearly, $B$ is positive, and by (\ref{adj}) we have that, for any positive
$R\in D\left( A_{0}\right) $,
\[
\int_{\Gamma _{0}}\left( L^{\dagger }R\right) \left( \eta \right) \lambda
\left( d\eta \right) =\int_{\Gamma _{0}}\left( L1\right) \left( \eta
\right) R\left( \eta \right) \lambda \left( d\eta \right) =0,
\]%
and hence,
\[
\int_{\Gamma _{0}}\left( BR\right) \left( \eta \right) \lambda \left(d
\eta \right) =\int_{\Gamma _{0}}\Xi\left( \eta \right) R\left( \eta
\right) \lambda \left(d \eta \right) .
\]
Now we apply Lemma~\ref{thTV-mult} with $P=B$ and $\rho=b_\alpha\geq1$, cf. (\ref{yy11}). Recall, that $\check{A}_{0}$ is
given by $\left( \check{A}_{0}R\right) \left( \eta \right) =-\Xi\left(
\eta \right) R\left( \eta \right) $ on the domain
\[
D(\check{A}_{0})=\left\{ R\in L^{1}\left( \Gamma _{0},b_\alpha \,d \lambda \right)
:\Xi R\in L^{1}\left( \Gamma _{0},b_\alpha \,d \lambda \right) \right\} .
\]%
Then, for any $0\leq R\in D(\check{A}_{0})$, we have%
\begin{eqnarray*}
&&\int_{\Gamma _{0}}\left( BR\right) \left( \eta \right) b_{\alpha}\left(
\eta \right) \lambda \left( d\eta \right) \\
&=&\int_{\Gamma _{0}}\left( L^{\dagger}R\right) \left( \eta \right) b_{\alpha}
\left( \eta \right) \lambda \left(d \eta \right) +\int_{\Gamma
_{0}}\Xi\left( \eta \right) R\left( \eta \right) b_{\alpha}\left( \eta
\right) \lambda \left( d\eta \right) \\
&=&\int_{\Gamma _{0}}R\left( \eta \right) \bigl( L b_{\alpha}\bigr) \left(
\eta \right) \lambda \left(d \eta \right) +\int_{\Gamma _{0}}\Xi\left(
\eta \right) R\left( \eta \right) b_{\alpha}\left( \eta \right) \lambda
\left(d \eta \right) ,
\end{eqnarray*}%
where we have used \eqref{adj} both sides of which are finite, see \eqref{main-est} below.
According to Lemma~\ref{thTV-mult}, we have to pick positive
$c$ and $\varepsilon$ such that
\begin{equation}\label{main-est}
\int_{\Gamma _{0}}\bigl( Lb_{\alpha}\bigr) \left( \eta \right) R\left( \eta
\right) \lambda \left( d\eta \right) \leq \int_{\Gamma _{0}}\left[ cb_{\alpha}
\left( \eta \right) -\varepsilon \Xi\left( \eta \right) \right] R\left(
\eta \right) \lambda \left(d \eta \right),
\end{equation}
holding for any positive $R\in D(\check{A}_{0})$.
By \eqref{deltainf} and \eqref{16A}, we get
\[
\bigl(Lb_{\alpha}\bigr)\left( \eta \right) =-\Xi (\eta)e^{\alpha \left\vert
\eta \right\vert }+e^{\alpha \left\vert
\eta \right\vert }e^{-\alpha }E\left( \eta \right) +e^{\alpha
\left\vert \eta \right\vert }e^{\alpha }\langle a^+\rangle |\eta|.
\]%
Hence, (\ref{main-est}) holds if, for ($\lambda$-almost) all $\eta\in \Gamma_0$, we have that
\[
e^{\alpha \left\vert \eta \right\vert }e^{-\alpha }E\left( \eta
\right) +e^{\alpha \left\vert \eta \right\vert }e^{\alpha }\langle a^+\rangle |\eta| \leq \left( c+\Xi\left(
\eta \right) \right) e^{\alpha \left\vert \eta \right\vert }-\varepsilon
\Xi\left( \eta \right),
\]%
which is equivalent to
\begin{equation}\label{exp-est}
\varepsilon \Xi\left( \eta \right) e^{-\alpha \left\vert \eta \right\vert }+(e^{\alpha }-1)\bigl(\langle a^+\rangle |\eta|-e^{-\alpha}E\left( \eta \right)\bigr) \leq c.
\end{equation}
For a given $\alpha>0$ and any $c>0$, by \eqref{yy2}, \eqref{16A}, \eqref{AB}, and \eqref{z11Aa},
it follows that
\[
\varepsilon \Xi\left( \eta \right) e^{-\alpha \left\vert \eta \right\vert }\leq c, \qquad \eta\in\Gamma_0,
\]
for some $\varepsilon>0$. Next, by \eqref{16A} the second term in the left-hand side of \eqref{exp-est} is non-positive whenever $\langle a^+\rangle \leq e^{-\alpha} m$, which holds for sufficiently small $\alpha>0$ in view of (\ref{yy12}).
\end{proof}

\subsection{Dual local evolution}

\label{SSec6.1}

Our aim now is to construct the evolution dual to that of $R^\Lambda \mapsto S_\alpha^\dagger (t) R^\Lambda$ obtained in
Lemma~\ref{y3tm}. Let $\mathcal{F}_\alpha$ be the dual space to $\mathcal{R}_{\alpha}$ as in (\ref{y1A}).
It is a weighted $L^\infty$ space on $\Gamma_0$
with measure $\lambda$ and norm
\begin{equation}
 \label{A7}
\|F\|_{\alpha} = \esssup_{\eta\in \Gamma_0} |F(\eta)| \exp(-\alpha |\eta|).
\end{equation}
Let $\widetilde{L}^\dagger_\alpha$ be the operator dual to $L^\dagger_\alpha = L^\dagger|_{\mathcal{R}_{\alpha}}$ as in Lemma~\ref{y3tm}.
Then the action of $\widetilde{L}^\dagger_\alpha$ is described in (\ref{R20}).
Let us show that, for any $\alpha' < \alpha$,
\begin{equation}
 \label{A8}
 \mathcal{F}_{\alpha'} \subset {\rm Dom} (\widetilde{L}^\dagger_\alpha).
\end{equation}
By (\ref{A7}) we have that $|F(\eta)| \leq \|F\|_{\alpha'} \exp(\alpha' |\eta|)$. Then
\begin{eqnarray*}
\| \widetilde{L}^\dagger_\alpha F \|_\alpha & \leq & \|F\|_{\alpha'} \esssup_{\eta\in \Gamma_0} \left(E(\eta) + \langle a^{+} \rangle |\eta| \right)
\exp\left( - (\alpha - \alpha')|\eta|\right) \\[.2cm] && + \, e^{-\alpha'} \|F\|_{\alpha'} \esssup_{\eta\in \Gamma_0} E(\eta)
\exp\left( - (\alpha - \alpha')|\eta|\right)\\[.2cm] && + \, e^{\alpha'} \|F\|_{\alpha'} \esssup_{\eta\in \Gamma_0}
\langle a^{+} \rangle |\eta| \exp\left( - (\alpha - \alpha')|\eta|\right),
\end{eqnarray*}
which can be rewritten in the form
\begin{equation}
\label{A9}
\| \widetilde{L}^\dagger_\alpha F \|_\alpha \leq \|F\|_{\alpha'}(1+ e^{\alpha'}) \Delta_+ (\alpha - \alpha') + \|F\|_{\alpha'}(1+ e^{-\alpha'})
 \Delta_- (\alpha - \alpha'),
\end{equation}
where, for $\beta >0$,
\begin{eqnarray*}
\Delta_+ (\beta) &:= & \esssup_{\eta\in \Gamma_0} \langle {a}_{+} \rangle |\eta| e^{-\beta |\eta|}, \\[.2cm]
\Delta_{-} (\beta) & := & \esssup_{\eta\in \Gamma_0} E(\eta)e^{-\beta |\eta|}. \nonumber
\end{eqnarray*}
Let $\mathcal{L}_\alpha$ stand for the closure of ${\rm Dom} (\widetilde{L}^\dagger_\alpha)$ in $\mathcal{F}_\alpha$. Note that
$\mathcal{L}_\alpha$ is a proper subspace of $\mathcal{F}_\alpha$. Set
\begin{equation*}
\mathcal{L}^{\odot}_\alpha =\{ F\in {\rm Dom} (\widetilde{L}^\dagger_\alpha ) : \widetilde{L}^\dagger_\alpha F \in \mathcal{L}_\alpha\}.
\end{equation*}
For $t\geq 0$, let $\widetilde{S}^{\odot}_\alpha(t)$ be the restriction, to $\mathcal{L}_\alpha$, of the operator dual to $S^\dagger_\alpha(t)$. By Theorem 10.4 in \cite{Pazy}, the operators $\widetilde{S}^{\odot}_\alpha(t)$, $t\geq 0$, constitute a $C_0$-semigroup on $\mathcal{L}_\alpha$, generated by $\widetilde{L}^\dagger_\alpha|_{\mathcal{L}^{\odot}_\alpha}$. The latter operator, which is the part of $\widetilde{L}^\dagger_\alpha$ in $\mathcal{L}_\alpha$, will be denoted by $\widetilde{L}^{\odot}_\alpha$. Note that, in view of (\ref{A8}) and (\ref{A9}), for any $\alpha'<\alpha$, the action of $\widetilde{L}^{\odot}_\alpha$ on $F\in \mathcal{F}_{\alpha'}$ is given by (\ref{R20}). Moreover, for any $\alpha'' < \alpha' < \alpha$, $\widetilde{L}^{\odot}_\alpha$ acts from $\mathcal{F}_{\alpha''}$ to $\mathcal{F}_{\alpha'}$, both considered as subsets of $\mathcal{L}^{\odot}_\alpha$.

For $\alpha'< \alpha$ and $F_0 \in \mathcal{F}_{\alpha'}$, we set
\begin{equation}
 \label{A12}
F_t = \widetilde{S}^{\odot}_\alpha(t) F_0, \qquad t>0.
\end{equation}
Then, see, e.g., page 5 in \cite{Pazy},
\begin{equation}
 \label{A13}
F_t = F_0 + \int_0^t \widetilde{L}^{\odot}_\alpha F_s d s.
\end{equation}

\subsection{The main statement}

We recall that any $k\in \mathcal{K}_{\alpha}$ is in fact a sequence of $k^{(n)}\in L^\infty ((\mathbb{R}^{d})^n)$, $n\in \mathbb{N}_0$, such that
\begin{equation*}
 \sup_{n\in \mathbb{N}}e^{\alpha n} \|k^{(n)}\|_{ L^\infty((\mathbb{R}^{d})^n)} < \infty, \qquad \alpha \in \mathbb{R},
\end{equation*}
see (\ref{z300}) and (\ref{z30}). By Proposition~\ref{rhopn}, such $k\in \mathcal{K}_{\alpha}$ is a correlation function of
a unique $\mu \in \mathcal{M}_{\rm fm}^1 (\Gamma)$ whenever $k^{(0)}=1$ and
\begin{equation}
 \label{y6}
 \langle \!\langle G,k \rangle \!\rangle \geq 0, \qquad \quad {\rm for} \ \ \ {\rm all} \ \ \ G\in B_{\rm bs}^+(\Gamma_0),
\end{equation}
see (\ref{9AY}) and (\ref{9AZ}). For $\alpha \in \mathbb{R}$, we set
\begin{equation*}
 \mathcal{M}_\alpha (\Gamma) = \{ \mu \in \mathcal{M}_{\rm fm}^1(\Gamma): k_\mu \in \mathcal{K}_\alpha\},
\end{equation*}
where $\mu$ and $k_\mu$ are as in (\ref{9AA}).
\begin{theorem}
 \label{r1tm}
Let $\theta$, $\alpha_*$, $\alpha^*$, and $T_*$ be as in Theorem~\ref{1tm}. Let also (\ref{yy12}) hold, and let $k_0 \in \mathcal{K}_{\alpha^*}$ be the correlation function of $\mu_0$, and $k_t$ be the solution of (\ref{20A}) with $k_t|_{t=0} = k_0$, as in Theorem~\ref{2tm}. Then, there exists $\mu_t\in \mathcal{M}_{\alpha_*}(\Gamma)$ such that $k_{\mu_t} = k_t$. In other words, the evolution $k_0 \mapsto k_t$ uniquely determines the evolution of the corresponding states
\[
\mathcal{M}_{\alpha^*} (\Gamma) \ni \mu_0 \mapsto \mu_t\in \mathcal{M}_{\alpha_*}(\Gamma), \qquad t>0.
\]
\end{theorem}
\begin{proof}
The main idea of the proof is to show that $k_t$ can be approximated in a certain sense by a sequence of `correlation functions', for which (\ref{y6}) holds. To this end we use two sequences
$\{\Lambda_n\}_{n\in \mathbb{N}} \subset \mathcal{B}_{\rm b}(\mathbb{R}^d)$
and $\{N_l\}_{l\in \mathbb{N}} \subset \mathbb{N}$. Both are increasing, and $\{\Lambda_n\}_{n\in \mathbb{N}}$ is exhausting, which means
that each $\Lambda\in \mathcal{B}_{\rm b}(\mathbb{R}^d)$ is contained in $\Lambda_n$ with big enough $n$.

Given $\mu_0 \in \mathcal{M}_{\alpha^*}$, let $k_0 \in \mathcal{K}_{\alpha^*}$ be such that $k_{\mu_0}= k_0$. Recall that this means that the projections $\mu^\Lambda$ are absolutely continuous with respect to $\lambda$, see (\ref{9AA}) and (\ref{y1}). For this $\mu_0$, and for $\Lambda_n$ and $N_l$ as above, we set
\begin{equation}
 \label{y8}
 R^{\Lambda_n, N_l}_0 (\eta) = R^{\Lambda_n}_0 (\eta) I_{N_l} (\eta),
\end{equation}
where $R^{\Lambda_n}_0$ is the local density as in (\ref{y1}), and
\begin{equation}
 \label{y9}
 I_{N} (\eta) := \left\{ \begin{array}{ll} 1, \quad \ &{\rm if} \ |\eta|\leq N;\\[.3cm] 0, \quad \ &{\rm otherwise}. \end{array}\right.
\end{equation}
Noteworthy, $R^{\Lambda_n, N_l}_0$ is a positive element of $L^1(\Gamma_0, d\lambda)$
with $\|R^{\Lambda_n, N_l}_0\|_{L^1(\Gamma_0, d\lambda)}\leq 1$,
and $R^{\Lambda_n, N_l}_0\in \mathcal{R}_{\alpha}$ for any $\alpha >0$. Indeed, cf. (\ref{y1A}),
\[
\Bigl\Vert R^{\Lambda_n, N_l}_0\Bigr\Vert_{\mathcal{R}_{\alpha}} = \sum_{m=0}^{N_l} \frac{r_m}{m!} e^{\alpha m},
\qquad r_m:= \Bigl\Vert\bigl(R^{\Lambda_n, N_l}_0\bigr)^{(m)}\Bigr\Vert_{L^1(\mathbb{R}^{md})}.
\]
Then, for any $\alpha >0$ and any $t\geq 0$, we can apply ${S}_\alpha^\dagger(t)$, as in Theorem~\ref{y3tm}, and obtain
\begin{equation}
 \label{y10}
 R^{\Lambda_n, N_l}_t = {S}^\dagger_\alpha(t) R^{\Lambda_n, N_l}_0 \in \mathcal{R}^{+}_{\alpha}:=
 \{ R\in \mathcal{R}_{\alpha} : R\geq 0\},
\end{equation}
which yields, cf. (\ref{A13}),
\begin{equation*}
 R^{\Lambda_n, N_l}_t = R^{\Lambda_n, N_l}_0 + \int_0^t {L}^\dagger_\alpha R^{\Lambda_n, N_l}_s ds .
\end{equation*}
For $G_0\in B_{\rm bs}^+(\Gamma_0)$, see (\ref{9AY}), let us consider
\begin{equation}
 \label{A16}
F_0 (\eta) = \sum_{\xi\subset \eta} G_{0} (\xi).
\end{equation}
Since $G_0(\xi) =0$ for all $\xi$ such that $|\xi|$ exceeds some $N(G_0)$, see (\ref{6A}), we have that
\begin{equation}
 \label{y120}
 |F_0 (\eta)| \leq (1+|\eta|)^{N(G_0)} C(G_0),
\end{equation}
for some $ C(G_0)>0$, and hence $F_0\in \mathcal{F}_\alpha$ for any $\alpha >0$. Therefore, the map $\mathcal{R}_{\alpha} \ni R \mapsto \langle \! \langle
 F_0 , R \rangle \! \rangle$ is continuous, and thus we can write, see (\ref{A9}),
 \begin{eqnarray}
  \label{y12}
\langle \! \langle F_0 ,R^{\Lambda_n}_t \rangle \! \rangle & = &  \langle \! \langle F_0 ,R^{\Lambda_n}_0 \rangle \! \rangle
+ \int_0^t  \langle \! \langle F_0 ,{L}^\dagger_\alpha R^{\Lambda_n}_s \rangle \! \rangle ds \\[.2cm]
& = &  \langle \! \langle F_0 ,R^{\Lambda_n}_0 \rangle \! \rangle + \int_0^t
\langle \! \langle \widetilde{L}^\dagger_\alpha F_0 ,R^{\Lambda_n}_s \rangle \! \rangle ds . \nonumber
\end{eqnarray}
Now we set, cf. (\ref{y2}),
\begin{equation}
 \label{y13}
 q^{\Lambda_n, N_l}_t (\eta) = \int_{\Gamma_0} R^{\Lambda_n, N_l}_t (\eta \cup \xi) \lambda (d \xi), \qquad t\geq 0.
\end{equation}
For $G_0\in B_{\rm bs}^+(\Gamma_0)$ and any $t\geq 0$, by (\ref{12A}) and (\ref{A16}) we have
\begin{eqnarray}
 \label{y14}
\langle \! \langle G_0 ,q^{\Lambda_n, N_l}_t \rangle \! \rangle & = &
\int_{\Gamma_0} \int_{\Gamma_0} G_0(\eta)R^{\Lambda_n, N_l}_t (\eta \cup \xi) \lambda (d\xi)
\lambda (d \eta) \\[.2cm]& = &
\int_{\Gamma_0} \biggl( \,\sum_{\xi \subset \eta} G_0 (\xi)\biggr) R^{\Lambda_n, N_l}_t (\eta )
 \lambda (d \eta)\nonumber \\[.2cm] & = &
 \langle \! \langle F_0 ,R^{\Lambda_n, N_l}_t \rangle \! \rangle , \nonumber
\end{eqnarray}
which in view of (\ref{9AY}) and (\ref{y10}) yields
\begin{equation}
 \label{y15}
\langle \! \langle G_0 ,q^{\Lambda_n, N_l}_t \rangle \! \rangle \geq 0.
\end{equation}
Applying again (\ref{12A}), for $\alpha >0$ we obtain, cf. (\ref{y14}),
\begin{eqnarray*}
\int_{\Gamma_0} e^{\alpha|\eta|} q^{\Lambda_n, N_l}_t (\eta) \lambda (d \eta) & = &
\int_{\Gamma_0} \biggl( \, \sum_{\xi \subset \eta} e^{\alpha |\xi|} \biggr)R^{\Lambda_n, N_l}_t (\eta ) \lambda (d \eta)\\[.2cm]
& = & \int_{\Gamma_0}\left( 1 + e^\alpha\right)^{|\eta|}R^{\Lambda_n, N_l}_t (\eta ) \lambda (d \eta).
\end{eqnarray*}
Since both $q^{\Lambda_n, N_l}_t$ and $R^{\Lambda_n, N_l}_t$ are positive and $R^{\Lambda_n, N_l}_t$ is in $\mathcal{R}_{\alpha'}$ for any $\alpha'>0$, the latter yields that, for any $\alpha>0$ and $t\geq 0$,
\begin{equation*}
q^{\Lambda_n, N_l}_t \in \mathcal{R}_{\alpha}.
\end{equation*}
As was already mentioned, our aim is to show that, in a weak sense, $q^{\Lambda_n, N_l}_t$ converges to $k_t$ as in Theorems~\ref{2tm}. Note that $k_t$ belongs to $\mathcal{K}_{\alpha_*}$, which is a completely different space than $\mathcal{R}_{\alpha_*}$, see (\ref{z300}).

To proceed further, we need to define the action of powers of $\widehat{L}$ as in (\ref{z})
on suitable sets of functions, which include $B_{\rm bs}(\Gamma_0)$. Recall that any function
$h: \Gamma_0 \to \mathbb{R}$ is a sequence of symmetric functions $h^{(n)} : (\mathbb{R}^{d})^n \to \mathbb{R}$, $n\in \mathbb{N}_0$, where
$h^{(0)}$ is a constant function. Let $\mathcal{H}_{\rm fin}$ be the set of measurable
$h: \Gamma_0 \to \mathbb{R}$, for each of which there exists $N(h)\in \mathbb{N}_0$ such that $h^{(n)} = 0$ whenever
$n> N(h)$. Then we set
\begin{eqnarray}
 \label{y17}
\mathcal{H}^1_{\rm fin} & = & \{h \in \mathcal{H}_{\rm fin}: h^{(n)} \in L^1 \left((\mathbb{R}^{d})^n\right), \quad {\rm for} \ \ n\leq N(h)\},\\[.2cm]
\mathcal{H}^\infty_{\rm fin} & = & \{h \in \mathcal{H}_{\rm fin}: h^{(n)} \in L^\infty \left((\mathbb{R}^{d})^n\right),
\quad {\rm for} \ \ n\leq N(h)\}.\nonumber
\end{eqnarray}
Note that
\begin{equation}
 \label{y18}
B_{\rm bs} (\Gamma_0) \subset \mathcal{H}^1_{\rm fin} \cap \mathcal{H}^\infty_{\rm fin},
\end{equation}
and, for any $\alpha >0$ and $\alpha' \in \mathbb{R}$,
\begin{equation}
 \label{y19}
\mathcal{H}^1_{\rm fin} \subset \mathcal{R}_{\alpha}, \qquad \mathcal{H}^\infty_{\rm fin} \subset \mathcal{K}_{\alpha'}.
\end{equation}
Furthermore, cf. (\ref{7A}) and (\ref{y120}), for any $\alpha>0$,
\begin{equation}
 \label{y190}
K : \mathcal{H}^\infty_{\rm fin} \to \mathcal{F}_\alpha.
\end{equation}
Let $A$ and $B$ be as in (\ref{z}) and (\ref{z1}), (\ref{z2}). Then, for $G\in \mathcal{H}^1_{\rm fin}\cap \mathcal{H}^\infty_{\rm fin}$ and $n\in \mathbb{N}_0$,
we have, see (\ref{14A}),
\begin{eqnarray*}
 \bigl\Vert \left( A G\right)^{(n)} \bigr\Vert_{L^\infty((\mathbb{R}^{d})^n)}
 & \leq & \left(nm + n^2 \|a^{-}\| \right) \bigl\Vert G^{(n)}\bigr\Vert_{L^\infty((\mathbb{R}^{d})^n)}\\[.2cm]
& & \ \qquad \qquad \qquad \qquad \qquad + n \langle a^{+} \rangle \bigl\Vert G^{(n+1)}\bigr\Vert_{L^\infty((\mathbb{R}^{d})^n)} , \\[.3cm]
 \bigl\Vert \left( B G\right)^{(n)} \bigr\Vert_{L^\infty((\mathbb{R}^{d})^n)} & \leq & n (n-1) \|a^{-}\|\bigl\Vert G^{(n+1)}\bigr\Vert _{L^\infty((\mathbb{R}^{d})^n)}\\[.2cm]
& & \ \qquad \qquad \qquad \qquad \qquad + n \langle a^{+} \rangle \bigl\Vert G^{(n)}\bigr\Vert _{L^\infty((\mathbb{R}^{d})^n)} , \\[.3cm]
 \bigl\Vert \left( A G\right)^{(n)} \bigr\Vert _{L^1((\mathbb{R}^{d})^n)}& \leq & \left(nm + n^2 \|a^{-}\| \right) \bigl\Vert G^{(n)}\bigr\Vert _{L^1((\mathbb{R}^{d})^n)} \\[.2cm]
& & \ \qquad \qquad \qquad \qquad \qquad + n \| a^{+} \| \bigl\Vert G^{(n+1)}\bigr\Vert_{L^1((\mathbb{R}^{d})^n)} , \\[.3cm]
\bigl\Vert \left( B G\right)^{(n)} \bigr\Vert_{L^1((\mathbb{R}^{d})^n)} & \leq & n \langle a^{-}
 \rangle \bigl\Vert G^{(n-1)}\bigr\Vert _{L^1((\mathbb{R}^{d})^n)}\\[.2cm]
 & & \ \qquad \qquad \qquad \qquad \qquad + n \| a^{+} \| \bigl\Vert G^{(n)}\bigr\Vert _{L^1((\mathbb{R}^{d})^n)} .
\end{eqnarray*}
Thus, $\widehat{L}$ given in (\ref{z}) can be defined on both sets (\ref{y17}) and
\begin{equation}
 \label{y20}
N({\widehat{L}G}) = N(G) + 1, \qquad \widehat{L}:\mathcal{H}^1_{\rm fin} \cap \mathcal{H}^\infty_{\rm fin}
 \to \mathcal{H}^1_{\rm fin} \cap \mathcal{H}^\infty_{\rm fin}.
\end{equation}
Now we fix $\Lambda_n$ and $N_l$, and let $\mu_0$ be in $\mathcal{M}_{\alpha^*}(\Gamma)$.
Then for $\eta \in \Gamma_{\Lambda_n}$, $k_0(\eta)$ is given by (\ref{9AA}), and hence, see (\ref{y13}), (\ref{y8}), and (\ref{y2}),
\begin{equation}
 \label{y170}
q^{\Lambda_n, N_l}_0 (\eta)
 \leq \int_{\Gamma_0} R_0^{\Lambda_n} (\eta \cup \xi) \lambda (d \xi) = k_0 (\eta), \quad \eta \in \Gamma_{\Lambda_n},
\end{equation}
which can readily be extended to all $\eta \in \Gamma_0$. Thus, $q^{\Lambda_n, N_l}_0 \in \mathcal{K}_{\alpha^*}$.
For $t\in [0, T_*)$, let $k_t^{\Lambda_n, N_l}$ be the solution of the problem (\ref{20A}) with
$k_t^{\Lambda_n, N_l}|_{t=0} = q_0^{\Lambda_n, N_l}$, as in Theorem~\ref{2tm}. Then, for $t\in (0,T_*)$,
\begin{equation*}
k_t^{\Lambda_n, N_l} = q_0^{\Lambda_n, N_l} + \int_0^t L^\Delta k_s^{\Lambda_n, N_l} ds,
\end{equation*}
where $L^\Delta$ is defined in (\ref{z33e}). In view of (\ref{y19}), for any
$G\in \mathcal{H}^1_{\rm fin} \cap \mathcal{H}^\infty_{\rm fin}$, we then have
\begin{equation}
 \label{y22}
\langle \! \langle G,k_t^{\Lambda_n, N_l} \rangle \! \rangle = \langle \! \langle G, q_0^{\Lambda_n, N_l} \rangle \! \rangle
 + \int_0^t \langle \! \langle \widehat{L} G,k_s^{\Lambda_n, N_l} \rangle \! \rangle ds.
\end{equation}
At the same time, for such $G$, we have that $KG$ is in each $\mathcal{F}_\alpha$, $\alpha >0$, cf. (\ref{y190}),
and hence, see (\ref{y12}) and (\ref{y14}),
\begin{equation*}
\langle \! \langle G,q_t^{\Lambda_n, N_l} \rangle \! \rangle = \langle \! \langle G, q_0^{\Lambda_n, N_l} \rangle \! \rangle
 + \int_0^t \langle \! \langle \widetilde{L}^\dagger_\alpha K G,R_s^{\Lambda_n, N_l} \rangle \! \rangle ds.
\end{equation*}
As was mentioned at the beginning of Subsection~\ref{SSec6.1}, the action of
$\widetilde{L}^\dagger_\alpha$ is described in (\ref{R20}). Thus, by (\ref{22}) we obtain from the latter
\begin{equation}
 \label{y24}
\langle \! \langle G,q_t^{\Lambda_n, N_l} \rangle \! \rangle = \langle \! \langle G, q_0^{\Lambda_n, N_l} \rangle \! \rangle
 + \int_0^t \langle \! \langle \widehat{L} G,q_s^{\Lambda_n, N_l} \rangle \! \rangle ds.
\end{equation}
For $G$ as in (\ref{y22}) and (\ref{y24}), we set
\begin{equation}
 \label{y25}
\phi (t, G) = \langle \! \langle G,k_t^{\Lambda_n, N_l} \rangle \! \rangle - \langle \! \langle G,q_t^{\Lambda_n, N_l} \rangle \! \rangle.
\end{equation}
Then
\begin{equation*}
\phi (t, G) = \int_0^t \phi (s, \widehat{L}G) ds, \qquad \phi (0, G) =0.
\end{equation*}
For any $n\in \mathbb{N}$, the latter yields
\begin{equation}
 \label{y27}
 \frac{d^n}{dt^t} \phi(t, G) = \phi (t, \widehat{L}^n G).
\end{equation}
In view of (\ref{y20}), $\phi (t,G)$ is infinitely differentiable on $(0,T_*)$, and
\[
\frac{d^n}{dt^t} \phi(0, G) =0, \qquad {\rm for} \ \ {\rm all} \ \ n
\in \mathbb{N}_0.
\]
Thus, $\phi (t,G)\equiv 0$, and hence, for all $G_0\in B_{\rm bs}^+(\Gamma_0)$, we have, see (\ref{y15}) and (\ref{y18}),
\begin{equation}
 \label{y28}
 \langle \! \langle G_0,q_t^{\Lambda_n, N_l} \rangle \! \rangle = \langle \! \langle G_0,k_t^{\Lambda_n, N_l} \rangle \! \rangle \geq 0.
\end{equation}
Now let $k_t$ be the solution of (\ref{20A}) with $k_t|_{t=0} = k_0$. In Appendix, we prove that, for any $G\in B_{\rm bs}(\Gamma_0)$,
\begin{equation}
 \label{y33}
 \langle \! \langle G,k_t \rangle \! \rangle = \lim_{n\to \infty} \lim_{l\to \infty} \langle \! \langle G,k_t^{\Lambda_n, N_l} \rangle \! \rangle,
\end{equation}
point-wise on $[0, T_*)$.
Then by (\ref{y28}) we get that, for each $t\in (0,T_*)$ and any $G\in B_{\rm bs}^+(\Gamma_0)$,
\[
\langle \! \langle G, k_t \rangle \! \rangle \geq 0,
\]
which together with the fact that $k_t\in \mathcal{K}_{\alpha^*}$ by Proposition~\ref{rhopn} yields
that $k_t$ is the correlation function for a certain unique $\mu_t \in
\mathcal{M}_{\alpha_*}(\Gamma)$.
\end{proof}
\begin{remark}
 \label{DRrk}
Theorem~\ref{r1tm} holds true for $m=0$ and $a^{+} \equiv 0$, which can be seen from (\ref{exp-est}).
\end{remark}

\begin{proposition}\label{newprop}
Let the conditions of Theorem~\ref{r1tm} hold.
Then $k_t$, as in Theorem~\ref{2tm}, and $k^D_t$, as in Theorem~\ref{3tm}, coincide for all $t\in [0,T_*)$, whenever $k^D_0 = k_0$.
\end{proposition}
\begin{proof}
As in the proof of Theorem~\ref{r1tm}, we are going to show that $k_t^D$ can be approximated by the same sequence of `correlation functions' (\ref{y13}). For $G_0\in B_{\rm bs}(\Gamma_0)$, let $F_0$ be as in (\ref{A16}). Since $F_0$ is polynomially bounded, see
(\ref{y120}), we have that $F_0 \in \mathcal{F}_{\alpha'}$ for any $\alpha' < \alpha$, where $\alpha$ is as in Theorem~\ref{y3tm}.
 Then we can apply (\ref{A12}) and obtain (\ref{A13}). For fixed $\Lambda_n$ and $N_l$, $R_0^{\Lambda_n, N_l}$
 is in any $\mathcal{R}_{\alpha'}$, and hence the map
 \[
 \mathcal{F}_\alpha \ni F \mapsto \langle \! \langle F, R_0^{\Lambda_n, N_l} \rangle \! \rangle \in \mathbb{R}
\]
is continuous. Since the Bochner integral in (\ref{A13}) is in $\mathcal{F}_\alpha$, we have
\begin{equation}
 \label{A15}
\langle \! \langle F_t , R_0^{\Lambda_n, N_l} \rangle \! \rangle = \langle \! \langle F_0 , R_0^{\Lambda_n, N_l}\rangle \! \rangle + \int_0^t
\langle \!\langle F_s , \widetilde{L}^\dagger_\alpha R_0^{\Lambda_n, N_l} \rangle \! \rangle ds.
\end{equation}
On the other hand,
\begin{equation*}
\widetilde{G}_t (\eta) := \sum_{\xi \subset \eta} (-1)^{|\eta \setminus \xi|} F_t(\xi)
\end{equation*}
is in $\mathcal{F}_\beta$, $\beta = \log ( 1 + e^\alpha)$. Thus, we can rewrite (\ref{A15}) in the form
\begin{equation}
 \label{A18}
\langle \! \langle \widetilde{G}_t , q_0^{\Lambda_n, N_l} \rangle \! \rangle = \langle \! \langle G_0 , q_0^{\Lambda_n, N_l} \rangle \! \rangle + \int_0^t
\langle \! \langle \widetilde{G}_s , {L}^\Delta q_0^{\Lambda_n, N_l} \rangle \! \rangle ds.
\end{equation}
For the evolution $G_0 \mapsto {G}_t\in\mathcal{G}_{\alpha^*}$ described by Theorem~\ref{1tm}, in a similar way we have
\begin{equation}
 \label{A19}
\langle \!\langle {G}_t , q_0^{\Lambda_n, N_l} \rangle \! \rangle = \langle \! \langle G_0 , q_0^{\Lambda_n, N_l}\rangle \! \rangle + \int_0^t
\langle \! \langle {G}_s , {L}^\Delta q_0^{\Lambda_n, N_l} \rangle \! \rangle ds.
\end{equation}
It is easy to see that $N(q_0^{\Lambda_n, N_l}) = N_l$ and, for any $m\leq N_l$, cf. (\ref{y170}),
\begin{eqnarray*}
\Bigl\Vert\bigl(q_0^{\Lambda_n, N_l}\bigr)^{(m)}\Bigr\Vert_{L^\infty(\mathbb{R}^{md})} & \leq & \bigl\Vert k_0^{(m)}\bigr\Vert_{L^\infty(\mathbb{R}^{md})},\\[.2cm]
\Bigl\Vert\bigl(q_0^{\Lambda_n, N_l}\bigr)^{(m)}\Bigr\Vert_{L^1(\mathbb{R}^{md})} & \leq & \sum_{s=0}^{N_l - m} \frac{1}{s!} \Bigl\Vert\bigl(R_0^{\Lambda_n}\bigr)^{(m+s)}\Bigr\Vert_{L^1(\Lambda^{m})}.
\end{eqnarray*}
Hence,
\begin{equation*}
q_0^{\Lambda_n, N_l} \in \mathcal{H}^1_{\rm fin} \cap \mathcal{H}^\infty_{\rm fin}.
\end{equation*}
Similarly as in (\ref{y20}), one can show that
\[
L^\Delta : \mathcal{H}^1_{\rm fin} \cap \mathcal{H}^\infty_{\rm fin} \to \mathcal{H}^1_{\rm fin} \cap \mathcal{H}^\infty_{\rm fin}.
\]
Now, for $h\in \mathcal{H}^1_{\rm fin} \cap \mathcal{H}^\infty_{\rm fin}$, we introduce, cf. (\ref{y25}),
\begin{equation*}
\phi(t, h) = \langle \! \langle \widetilde{G}_t , h \rangle \! \rangle - \langle \! \langle {G}_t , h \rangle \! \rangle,
\end{equation*}
for which by (\ref{A19}) and (\ref{A18}) we get
\begin{equation*}
\phi(t, h) = \int_0^t \phi(s, L^\Delta h) ds, \qquad \phi(0,h) =0.
\end{equation*}
Employing the same arguments as in (\ref{y27}), (\ref{y28}) we then obtain
\begin{equation}
 \label{y37}
 \langle \!\langle \widetilde{G}_t , q_0^{\Lambda_n, N_l} \rangle \! \rangle = \langle \!\langle {G}_t , q_0^{\Lambda_n, N_l} \rangle \! \rangle.
\end{equation}
On the other hand, by (\ref{z46}) we have
\begin{equation*}
 \langle \!\langle {G}_t , q_0^{\Lambda_n, N_l} \rangle \! \rangle = \langle \!\langle {G}_0 , \tilde{k}_t^{\Lambda_n, N_l} \rangle \! \rangle,
\end{equation*}
where the evolution $q_0^{\Lambda_n, N_l} \mapsto \tilde{k}_t^{\Lambda_n, N_l}$ is described by Theorem~\ref{3tm}.
At the same time, see (\ref{y14}),
\begin{eqnarray*}
 \langle \!\langle \widetilde{G}_t , q_0^{\Lambda_n, N_l} \rangle \! \rangle & = &
  \langle \!\langle F_t , R_0^{\Lambda_n, N_l} \rangle \! \rangle = \langle \!\langle F_0 , R_t^{\Lambda_n, N_l} \rangle \! \rangle \\[.2cm]
  & = & \langle \!\langle {G}_0 , q_t^{\Lambda_n, N_l} \rangle \! \rangle , \nonumber
\end{eqnarray*}
where $q_t^{\Lambda_n, N_l}$ is the same as in (\ref{y13}) and (\ref{y28}). Then (\ref{y37}) can be rewritten
\begin{equation*}
 \langle \!\langle {G}_0 , \tilde{k}_t^{\Lambda_n, N_l} \rangle \! \rangle = \langle \!\langle {G}_0 , q_t^{\Lambda_n, N_l} \rangle \! \rangle ,
\end{equation*}
which holds for all $G_0 \in B_{\rm bs}(\Gamma_0)$. Then, by (\ref{y28}) we have that, for all $G_0 \in B_{\rm bs}(\Gamma_0)$,
\begin{equation*}
 \langle \!\langle {G}_0 , \tilde{k}_t^{\Lambda_n, N_l} \rangle \! \rangle = \langle \!\langle {G}_0 , k_t^{\Lambda_n, N_l} \rangle \! \rangle ,
\end{equation*}
 and, for $G_0 \in B^{+}_{\rm bs}(\Gamma_0)$,
\begin{equation*}
 \langle \!\langle {G}_0 , \tilde{k}_t^{\Lambda_n, N_l} \rangle \! \rangle \geq 0.
\end{equation*}
At the same time, by (\ref{z46}) we have
\begin{eqnarray*}
& & \langle \! \langle G_0,{k}^D_t \rangle \! \rangle - \langle \! \langle G_0,\tilde{k}_t^{\Lambda_n, N_l} \rangle \! \rangle =
\langle \! \langle G_t,k_0 \rangle \! \rangle - \langle \! \langle G_t,q_0^{\Lambda_n, N_l} \rangle \! \rangle \\[.2cm] & & \qquad
= \int_{\Gamma_0} G_t(\eta) k_0 (\eta) \bigl( 1 - \mathbb{I}_{\Gamma_{\Lambda_n}} (\eta)\bigr)\lambda (d \eta) \nonumber
\\[.2cm] & & \qquad\qquad +\int_{\Gamma_0}G_t(\eta) \Bigl[ k_0 (\eta)\mathbb{I}_{\Gamma_{\Lambda_n}} (\eta) - q_0^{\Lambda_n, N_l} (\eta)\Bigr]\lambda (d \eta). \nonumber
\end{eqnarray*}
Then exactly as in (\ref{y33}) we obtain
\begin{equation*}
\langle \! \langle G_0,k^D_t \rangle \! \rangle = \lim_{n\to \infty} \lim_{l\to \infty}
\langle \! \langle G_0,q_t^{\Lambda_n, N_l} \rangle \! \rangle
= \lim_{n\to \infty} \lim_{l\to \infty}
\langle \! \langle G_0,\tilde{k}_t^{\Lambda_n, N_l} \rangle \! \rangle,
\end{equation*}
which holds for any $G_0 \in B_{\rm bs}(\Gamma_0)$. Thus, for all $G_0 \in B^{+}_{\rm bs}(\Gamma_0)$,
\begin{eqnarray*}
& {\rm (a)} & \forall G_0 \in B_{\rm bs}(\Gamma_0) \quad \langle \! \langle G_0,k^D_t \rangle \! \rangle
= \langle \! \langle G_0,k_t \rangle \! \rangle, \\[.2cm]
& {\rm (b)} & \forall G_0 \in B^{+}_{\rm bs}(\Gamma_0) \quad
\langle \! \langle G_0,k^D_t \rangle \! \rangle \geq0.
\end{eqnarray*}
The latter property yields that $k^D_t$ is a correlation function. To complete the proof we have to show that (a) implies $k_t= k^D_t$. In the topology induced from $\Gamma$, each $\Gamma_\Lambda$, $\Lambda \in \mathcal{B}_{\rm b}(\mathbb{R}^d)$, is a Polish space. Let $\mathcal{C}_\Lambda$ be the set of all bounded continuous $G: \Gamma_\Lambda \to \mathbb{R}$. Since (a) holds for all $G\in \mathcal{C}_\Lambda\cap B_{\rm bs}(\Gamma_0)$, the projections of the correlation measures (\ref{9A}) corresponding to $k_t$ and $k^D_t$ on each $\mathcal{B}(\Gamma_\Lambda)$ coincide, see, e.g. Proposition 1.3.27 in \cite{A}. This yields $k_t=k^D_t$, and hence completes the proof.
 \end{proof}

\section{Concluding Remarks}
\label{SCon}

In Theorem~\ref{1tm}, we have shown that the evolution of quasi-observables exists
with the only condition that (\ref{z14}) holds. However, this evolution is restricted in time
and takes place in the scale of Banach spaces (\ref{z3}), (\ref{z5}).
In \cite{Dima}, the analytic semigroup that defines the evolution of quasi-observables was constructed.
Thus, the evolution $G_0 \mapsto G_t$ defined by this semigroup takes place in one space $\mathcal{G}_\alpha$ for all $t>0$.
However, this result was obtained under the additional condition that in our notations takes the form
\begin{equation}
 \label{Conc}
m > 4 \bigl(\langle a^{-} \rangle e^{-\alpha} + \langle a^{+} \rangle\bigr),
\end{equation}
by which the constant mortality should dominate not only the dispersion but also the competition.
In our Theorem~\ref{1tm} the value of $m$ can be arbitrary and even equal to zero.

As was shown in \cite{DimaN2}, under conditions similar to (\ref{z14}) and (\ref{Conc})
the corresponding semigroup evolution $k_0 \mapsto k_t$ exists in a proper Banach subspace of $\mathcal{K}_\alpha$, cf. (\ref{z300}).
In our Theorem~\ref{2tm}, we construct the evolution $k_0 \mapsto k_t$ in the scale of spaces (\ref{z31}), restricted in time,
but under the condition of (\ref{z14}) only. Hence, it holds for any $m\geq 0$.
Note that the problem of whether $k_t$ for $t>0$ are correlation functions of probability measures on $\mathcal{B}(\Gamma)$,
provided $k_0$ is, has not been studied yet in the literature. In our Theorem~\ref{r1tm}, we prove that this
evolution corresponds to the evolution of probability measures if
$m > \langle a^{+} \rangle$ holds in addition to (\ref{z14}), cf.~Remark~\ref{DRrk}.

\appendix

\section{Proof of \eqref{y33}}

For fixed $t\in (0,T_*)$ and $G_0\in B_{\rm bs}(\Gamma_0)$, by (\ref{z47}) we have
\begin{equation}
 \label{y28a}
 \langle \! \langle G_0,k_t \rangle \! \rangle - \langle \! \langle G_0,k_t^{\Lambda_n, N_l} \rangle \! \rangle = \langle \! \langle G_t^D,k_0 \rangle \! \rangle - \langle \! \langle G_t^D,q_0^{\Lambda_n, N_l} \rangle \! \rangle = \mathcal{I}_n^{(1)}+
 \mathcal{I}_{n,l}^{(2)},
\end{equation}
where we set
\begin{eqnarray*}
 \mathcal{I}_n^{(1)}
 &= & \int_{\Gamma_0} G_t^D(\eta) k_0 (\eta) \bigl( 1 - \mathbb{I}_{\Gamma_{\Lambda_n}} (\eta)\bigr)\lambda (d \eta),
\\[.2cm] \mathcal{I}_{n,l}^{(2)} & = & \int_{\Gamma_0}G_t^D(\eta) \Bigl[ k_0 (\eta)\mathbb{I}_{\Gamma_{\Lambda_n}} (\eta) - q_0^{\Lambda_n, N_l} (\eta)\Bigr]\lambda (d \eta). \nonumber
\end{eqnarray*}
Let us prove that, for an arbitrary $\varepsilon >0$,
\begin{equation}
 \label{y29a}
 |\mathcal{I}_n^{(1)}|< \varepsilon/2,
\end{equation}
for sufficiently big $\Lambda_n$. Recall that $k_0$ is a correlation function, and hence is positive.
Taking into account that
\[
 \mathbb{I}_{\Gamma_\Lambda} (\eta) = \prod_{x\in \eta} \mathbb{I}_\Lambda (x),
\]
we have
\begin{eqnarray}
\label{Ap1}
|\mathcal{I}_n^{(1)}| & \leq & \int_{\Gamma_0} \left\vert G^D_t(\eta)\right\vert k_0 (\eta)
(1- \mathbb{I}_{\Gamma_{\Lambda_n}}(\eta)) \lambda (d \eta) \\[.2cm] & = & \sum_{p=1}^\infty \frac{1}{p!}
 \int_{(\mathbb{R}^d)^p}\left\vert(G_t^D)^{(p)} (x_1, \dots x_p)\right\vert k_0^{(p)}(x_1, \dots x_p)\nonumber \\[.2cm]
&& \times J_{\Lambda_n} (x_1, \dots x_p) d x_1 \cdots d x_p,\nonumber
\end{eqnarray}
where
\begin{eqnarray}
 J_\Lambda (x_1 , \dots , x_p) &:=& 1 - \mathbb{I}_\Lambda (x_1) \cdots \mathbb{I}_\Lambda (x_p) \nonumber\\[.2cm]
&=&  \mathbb{I}_{\Lambda^c} (x_1) \mathbb{I}_\Lambda (x_2)\cdots \mathbb{I}_\Lambda (x_p) +
\mathbb{I}_{\Lambda^c} (x_2) \mathbb{I}_\Lambda (x_3)\cdots \mathbb{I}_\Lambda (x_p) \nonumber \\[.2cm]& & +
\cdots + \mathbb{I}_{\Lambda^c} (x_{p-1}) \mathbb{I}_{\Lambda} (x_p) + \mathbb{I}_{\Lambda^c} (x_p),
 \nonumber \\[.2cm] & \leq& \sum_{s=1}^p \mathbb{I}_{\Lambda^c} (x_{s}), \label{Ap2}
\end{eqnarray}
and $\Lambda^c := \mathbb{R}^d \setminus \Lambda$.
Taking into account that $k_0\in \mathcal{K}_{\alpha^*}$, cf. (\ref{z300}) and (\ref{z30}), by (\ref{Ap2}) we obtain in (\ref{Ap1})
\begin{equation}
 \label{Ap3}
 |\mathcal{I}_n^{(1)}| \leq \|k_0\|_{\alpha^*} \sum_{p=1}^\infty \frac{p}{p!} e^{-\alpha^* p} \int_{\Lambda_n^c} \int_{(\mathbb{R}^d)^{p-1}}
\Bigl\vert(G_t^D)^{(p)} (x_1, \dots x_p)\Bigr\vert d x_1 \cdots d x_p.
\end{equation}
For $t$ as in (\ref{y28a}), one finds $\alpha< \alpha^*$ such that
$G_t^D \in \mathcal{G}_{\alpha}$, see Remark~\ref{Grk}. For this $\alpha$ and $\varepsilon$ as in (\ref{y29a}), we pick $\bar{p}\in \mathbb{N}$
such that
\begin{equation}
\label{Ap4}
\sum_{p=\bar{p}+1}^\infty \frac{e^{-\alpha p}}{p!}
 \int_{(\mathbb{R}^d)^p}\left\vert(G_t^D)^{(p)} (x_1, \dots x_p)\right\vert d x_1 \cdots d x_p <  \frac{ \varepsilon e(\alpha^* - \alpha)}{ 4 \|k_0\|_{\alpha^*}}.
\end{equation}
Then we apply (\ref{Ap4}) and the following evident estimate
\[
p e^{-\alpha^* p} \leq e^{-\alpha p} / e(\alpha^* - \alpha),
\]
and obtain in (\ref{Ap3}) the following
\begin{eqnarray*}
|\mathcal{I}_n^{(1)}| & < & \frac{\|k_0\|_{\alpha^*}}{e(\alpha^* - \alpha)} \sum_{p=1}^{\bar{p}} \frac{p}{p!} e^{-\alpha^* p} \int_{\Lambda_n^c} \int_{(\mathbb{R}^d)^{p-1}}
\left\vert(G_t^D)^{(p)} (x_1, \dots x_p)\right\vert d x_1 \cdots d x_p + \frac{\varepsilon}{4}.
\end{eqnarray*}
Here the first term contains a finite number of summands, in each of which $(G_t^D)^{(p)} $ is in $L^1 ((\mathbb{R}^d)^p)$.
Hence, it can be made strictly smaller than $\varepsilon/4$ by picking big enough
$\Lambda_n$, which yields (\ref{y29a}).

Let us show the same for the second integral in (\ref{y28a}). Write, see (\ref{y2}), (\ref{y13}), (\ref{y8}), and (\ref{y9}),
\begin{eqnarray*}
\mathcal{I}^{(2)}_{n,l} & = & \int_{\Gamma_0} G_t^D(\eta) \int_{\Gamma_0} R^{\Lambda_n}_0 (\eta \cup \xi) \left[1- I_{N_l} (\eta \cup \xi)\right]\lambda (d\eta) \lambda (d \xi)\\ & = & \int_{\Gamma_0} F_t(\eta)R^{\Lambda_n}_0 (\eta ) \left[1- I_{N_l} (\eta)\right]\lambda (d\eta) \nonumber \\[.2cm]
& = & \sum_{m=N_l+1}^\infty \frac{1}{m!} \int_{\Lambda_n^m}\left( R^{\Lambda_n}_0 \right)^{(m)} (x_1 , \dots , x_m) F^{(m)}_t (x_1 , \dots , x_m)
dx_1 \cdots dx_m, \nonumber
\end{eqnarray*}
where
\[
F_t(\eta) := \sum_{\xi\subset \eta} G_t^D (\xi),
\]
and hence
\begin{equation}
 \label{y31}
F^{(m)}_t (x_1 , \dots , x_m)= \sum_{s=0}^m \ \sum_{\{i_1, \dots , i_s\} \subset \{1, \dots , m\}} \bigl( G_t^D\bigr)^{(s)}(x_{i_1}, \dots , x_{i_s}).
\end{equation}
By (\ref{y2}), for $x_i\in \Lambda_n$, $i=1, \dots , m$, we have
\[
k_0^{(m)}(x_1 , \dots , x_m) = \sum_{s=0}^\infty \int_{\Lambda_n^s} \bigl( R^{\Lambda_n}_0 \bigr)^{(m+s)} (x_1 , \dots , x_m, y_1, \dots y_s) dy_1\cdots d y_s,
\]
from which we immediately get that
\begin{equation*}
\bigl( R^{\Lambda_n}_0 \bigr)^{(m)} (x_1 , \dots , x_m) \leq k_0^{(m)}(x_1 , \dots , x_m) \leq e^{-\alpha^*m}\|k_0\|_{\alpha^*},
\end{equation*}
since $k_0 \in \mathcal{K}_{\alpha^*}$. Now let $\Lambda_n$ be such that (\ref{y29a}) holds.
Then we can have
\begin{equation}\label{y32a}
|\mathcal{I}^{(2)}_{n,l}| < \varepsilon /2,
\end{equation}
holding for big enough $N_l$ if
$e^{-\alpha^* |\cdot|} F_t$ is in $L^1 (\Lambda_n, d \lambda)$. By (\ref{y31}),
\begin{eqnarray*}
& & \sum_{p=0}^\infty \frac{1}{p!} e^{-\alpha^* p} \int_{\Lambda_n^p} |F^{(p)} (x_1 , \dots , x_p)| dx_1 \cdots dx_p \\[.2cm]
& & \qquad \leq \sum_{p=0}^\infty
 \sum_{s=0}^p \frac{1}{s! (p-s)!} e^{-\alpha^* s} \Bigl\Vert\bigl( G_t^D\bigr)^{(s)}\Bigr\Vert_{L^1(\Gamma_0, d\lambda)}e^{-\alpha^* (p-s)} [\ell(\Lambda_n)]^{p-s}\\[.2cm]
& & \qquad = \|G^D_t\|_{\alpha^*} \exp\Bigl( e^{-\alpha^*}\ell(\Lambda_n) \Bigr),
\end{eqnarray*}
where $\ell(\Lambda_n)$ is the Lebesgue measure of $\Lambda_n$. This yields (\ref{y32a}) and thereby also (\ref{y33}).

\section*{Acknowledgment}
This paper is partially supported by the DFG through the SFB 701
``Spektrale Strukturen und Topologische Methoden in der Mathematik''
and through the research project 436 POL 113/125/0-1.

\end{document}